\documentclass[journal]{IEEEtran}

\usepackage{array}
\usepackage{booktabs}
\usepackage{graphicx}
\usepackage{amsmath}
\usepackage{amsfonts}
\usepackage{amssymb}
\usepackage{amsthm}
\usepackage{multirow}
\usepackage{anyfontsize}
\usepackage{cite}
\usepackage{xcolor}

\usepackage{blindtext}
\newcommand\blfootnote[1]{%
\begingroup
\renewcommand\thefootnote{}\footnote{#1}%
\addtocounter{footnote}{-1}%
\endgroup
}

\usepackage{amsfonts}

\newtheorem{lemma}{Lemma}

\usepackage[linesnumbered,algoruled,boxed,lined]{algorithm2e}

\SetAlCapNameFnt{\small}
\SetAlCapFnt{\small}

\DeclareMathOperator{\diag}{diag}

\ifCLASSINFOpdf

\else

\fi

\begin{document}

\setlength{\textfloatsep}{2pt}

\title{Cluster Index Modulation for Reconfigurable Intelligent Surface-Assisted mmWave Massive MIMO
\vspace{-0.25ex}
}

\author{Mahmoud~Raeisi,~\IEEEmembership{Student Member,~IEEE},
        Asil~Koc,~\IEEEmembership{Graduate Student Member,~IEEE},
        Ibrahim~Yildirim,~\IEEEmembership{Graduate Student Member,~IEEE},
        Ertugrul~Basar,~\IEEEmembership{Fellow,~IEEE}, and
        Tho Le-Ngoc,~\IEEEmembership{Life Fellow,~IEEE
        \vspace{-7ex}
        }
\thanks{Manuscript received 23 July 2022; revised 2 December 2022 and 15 April 2023; accepted 15 June 2023. This work was supported by the Scientific and Technological Research Council of Turkey (TÜBİTAK) under Grants 120E401 and BİDEB 2214. The associate editor coordinating the review of this article and approving it for publication was Dr. Daniel So. (\textit{Corresponding author: Ertugrul Basar.})}
\thanks{M. Raeisi, I. Yildirim, and E. Basar are with the Communications Research and Innovation Laboratory (CoreLab), Department of Electrical and Electronics Engineering, Koç University, Sariyer, Istanbul 34450, Turkey. (e-mail: mraeisi19@ku.edu.tr; ibrahimyildirim19@ku.edu.tr; ebasar@ku.edu.tr)}
\thanks{I. Yildirim is also with the Faculty of Electrical and Electronics Engineering, Istanbul Technical University, Istanbul 34469, Turkey. (e-mail: yildirimib@itu.edu.tr)}
\thanks{I. Yildirim, A. Koc, and T. Le-Ngoc are with Department of Electrical and Computer Engineering, McGill University, Montreal, QC, Canada. (e-mail: ibrahim.yildirim@mail.mcgill.ca; asil.koc@mail.mcgill.ca; tho.le-ngoc@mcgill.ca)}

}

\markboth{IEEE TRANSACTIONS ON WIRELESS COMMUNICATIONS}
{Shell \MakeLowercase{\textit{et al.}}: Bare Demo of IEEEtran.cls for IEEE Journals}

\maketitle

\begin{abstract}
In this paper, we propose a transmission mechanism for a reconfigurable intelligent surface (RIS)-assisted millimeter wave (mmWave) system based on cluster index modulation (CIM), named best-gain optimized cluster selection CIM (BGCS-CIM).
The proposed BGCS-CIM scheme considers effective cluster power gain and spatial diversity gain obtained by the additional paths within the indexed cluster to construct an efficient codebook. We also integrate the proposed scheme into a practical system model to create a virtual path between transmitter and receiver where the direct link has been blocked. Thanks to the designed whitening filter, a closed-form expression for the upper bound on the average bit error rate (ABER) is derived and used to validate the simulation results. It has been shown that the proposed BGCS-CIM scheme outperforms the existing benchmarks thanks to its higher effective cluster gain, spatial diversity of indexed clusters, and lower inter-cluster interference.
\vspace{-1ex}
\end{abstract}


\IEEEpeerreviewmaketitle

\vspace{-4ex}
\section{Introduction}
\IEEEPARstart{T}{he} intriguing applications of next-generation wireless networks demand more bandwidth, which is already scarce in the sub-6 GHz communication band \cite{9386246}. A promising solution is adopting the millimeter wave (mmWave) band, which spans from 30 GHz to 300 GHz \cite{9386246,wang2020joint}. On the other hand, mmWave signals suffer from high path loss and attenuation. Thanks to the short wavelength of mmWave signals, it is possible to pack a large antenna array and form a high directivity beam to compensate for the severe path loss  \cite{wang2020joint,wang2020intelligent}. Nevertheless, high directivity beams are vulnerable to blockage, especially in a dense urban area or indoor environment. Reconfigurable intelligent surfaces (RISs) can be installed in the environment to handle the blockage problem and increase the coverage area by providing additional transmission paths \cite{wang2020joint, wang2020intelligent, ying2020gmd, basar2019wireless, basar2021reconfigurable}. Another problem encountered in mmWave transmission compared to sub-6 GHz systems is its higher hardware cost and power consumption \cite{guo2021joint, 9687593}. These drawbacks necessitate wireless researchers to consider intelligent and efficient schemes that decrease the cost and power consumption. Moreover, index modulation (IM) is considered a promising solution to obtain a lower hardware cost with a reduced number of radio-frequency (RF) chains and enhanced spectral efficiency performance by transmitting extra information bits \cite{jiang2021generalized, basar2016index, basar2017index}. In IM systems, additional information bits can be transmitted by indexing different building blocks such as antenna elements, subcarriers, paths, and clusters \cite{jiang2021generalized, ding2017spatial, raeisi2022cluster, liu2020machine}.

\vspace{-2.5ex}
\subsection{Related works}
\vspace{-0.75ex}
RIS-empowered communication can be considered a promising technology to extend the coverage area in mmWave communication systems. The optimal placement for RIS implementation in mmWave systems has been investigated in \cite{9386246, yildirim2020modeling}. It is shown that an increased achievable rate is obtained when an RIS is located in the proximity of the transmitter or receiver due to the multiplicative path loss of the RIS-assisted path. The authors in \cite{wang2020intelligent} proposed a framework to extend the coverage area by adopting several RISs to assist communication from the base station to the single antenna user. Specifically, they proposed a joint active and passive precoding design to maximize the received signal power. By introducing such a scheme, the proposed RIS-assisted scheme in \cite{wang2020intelligent} can effectively increase the robustness of the mmWave system against blockage. In \cite{wang2020joint}, the authors investigated a mmWave communication system that adopted large antenna arrays at both transmitter and receiver. An optimization problem is formulated to jointly adjust reflection coefficients at the RIS and the hybrid precoder/combiner at the transmitter/receiver to maximize spectral efficiency. Guo \textit{et al.} study a power minimization problem in \cite{guo2021joint} by jointly optimizing digital and analog beamforming matrices at the transmitter and reflection coefficients at the RIS under a given signal-to-interference-plus-noise-ratio (SINR) constraint. 

Spatial scattering modulation (SSM) is introduced in \cite{ding2017spatial} by indexing a set of orthogonal paths in the mmWave environment. However, the authors assume that there is only one cluster in the environment in  \cite{ding2017spatial}, this approach might be impractical for most the mmWave transmission environment. Another drawback of the SSM scheme is the orthogonality assumption among paths. Generalized SSM (GSSM) is also proposed in  \cite{tu2018generalized} by indexing a set of paths instead of a single path. Both works of \cite{ding2017spatial} and \cite{tu2018generalized} adopt linear arrays at the transmitter and receiver, which are not capable of resolving scattering paths in both azimuth and elevation directions. Moreover, in \cite{jiang2021generalized}, the authors investigated a generalization three dimensions (3D) SSM scheme to facilitate the problem of path resolution in both azimuth and elevation dimensions. They also proposed a whitening filter to deal with correlated noises at the receiver. Against this background, the authors of the current work recently introduced cluster index modulation (CIM) in \cite{raeisi2022cluster} in order to efficiently index the clusters in the environment. The well-known Saleh-Valenzuela channel model with multiple clusters is assumed in the CIM scheme, and scattering clusters in the environment are indexed instead of paths. It has been demonstrated in \cite{raeisi2022cluster} that under non-orthogonal path conditions, SSM could not perform well because the paths are generally close to each other, which results in interference on the other indexed paths. Since there is a considerable physical distance between the clusters, indexing clusters is a more reasonable idea, and the superiority of CIM over SSM is shown in terms of bit error rate (BER) in \cite{raeisi2022cluster}.

More recently, \cite{gopi2020intelligent} introduced a twin-RIS structure to perform a beam-index modulation scheme using two RISs between the terminals. The first RIS is located close to the transmitter, while the second one is close to the receiver. The reflecting elements on the second RIS are indexed for the IM purpose and reflect the received beam to the receiver. According to the IM symbol, the first RIS adjusts the phases of the received beam from the transmitter to ensure that only the target element on the second RIS receives the beam.

\vspace{-4ex}
\subsection{Motivations and contributions}
\vspace{-0.8ex}
The state-of-the-art of IM schemes in mmWave communications assumed orthogonality among paths when the array size is large \cite{jiang2021generalized,ding2017spatial,li2019polarized,tu2018generalized, ruan2019diversity}. Nonetheless, this orthogonality assumption is unrealistic, and in practice, there is a correlation among each couple of paths that is not ignorable for indexing paths \cite{raeisi2022cluster}. In the realistic channel model, which includes correlation among paths, power leakage into undesired paths can result in interference, making IM detection more difficult. Hence, a practical IM scheme should be appropriately designed to prevent interference between the paths and even exploit this leakage to enhance the system performance. Although \cite{raeisi2022cluster} demonstrates the superiority of the CIM scheme compared to traditional SSM, the performance of the CIM scheme under the situation that clusters are randomly distributed in both azimuth and elevation dimensions has not been studied before.

As we mentioned earlier, the mmWave channel is vulnerable to blockage due to the intrinsic characteristics of the mmWave signals. Adopting an RIS to the transmission can potentially combat this blockage problem \cite{wang2020joint, wang2020intelligent, ying2020gmd, basar2019wireless, basar2021reconfigurable}. Incorporating IM into RIS-assisted mmWave communications systems is also interesting due to its potential to enhance the coverage, energy efficiency, achievable rate performance, and to decrease the overall cost. From this point of view, an RIS-assisted mmWave communications system model should be designed delicately to make the IM effective. Accordingly, the location of terminals and the IM entities should be selected precisely. In the existing literature, there is a need for a comprehensive IM-based RIS-assisted scheme that considers all of these concerns.

\begin{table}
\caption{List of frequently-used mathematical symbols.\label{Tab:MathSymbols}}
\begin{tabular}{c|p{6.5cm}} 
 \toprule
 $\mathbf{G}$ & S-V channel between Tx and RIS \\ 

 $\mathbf{R}$ & S-V channel between RIS and Rx  \\

 $N_i$ & Number of antenna elements at $i \in$ $\{$Tx, Rx$\}$  \\
 
 $N$ & Number of reflectors at RIS  \\
 
 $G_t$ ($G_r$) & Transmitter (receiver) antenna array gain \\
 
 $C_i$ ($L_i$) & Number of clusters (paths), where $i \in \{ G, R \}$\\
 
 $\alpha$ ($\beta$) & Complex channel gain of $\mathbf{G}$ ($\mathbf{R}$)\\
 
 $\phi^t$ ($\phi^r$) & Azimuth angle of departure (arrival) associated with the RIS\\
 
 $\theta^t$ ($\theta^r$) & Elevation angle of departure (arrival) associated with the RIS\\
 
 $\varphi^t$ ($\varphi^r$) & Azimuth angle of departure (arrival) associated with the transmitter (receiver)\\
 
 $\vartheta^t$ ($\vartheta^r$) & Elevation angle of departure (arrival) associated with the transmitter (receiver)\\
 
 $\mathbf{a}_i(.) $ & Antenna array response associated with the $i \in \{$Tx, Rx$\}$. \\

 $\mathbf{a}_{i,{ris}}(.)$ & $i \in \{$ Transmit, Receive$\}$ antenna array response associated with the RIS\\
 
 $P$ & Transmit power \\
 
 $\mathbf{\Psi}$ & RIS phase shift matrix \\
 
 $\mathbf{\mathbf{f}_t}$ & Analog beamformer vector at the Tx \\
 
 $\mathbf{\mathbf{W}}$ & Analog combiner matrix \\
 
 $\mathbf{y}$ & Receive signal vector before combiner \\
 
 $\mathbf{z}$ & Receive signal vector after combiner \\
 
  $\mathbf{z_F}$ & Filtered receive signal with whitened noise \\
 
 $\mathcal{B}$ & CIM codebook\\
 
 $B$ & CIM codebook order\\
 
 \bottomrule
\end{tabular}
\end{table}

Based on the aforementioned gaps and main motivation, our contributions in this paper are summarized as follows:

\begin{itemize}
    \item This paper proposes an advanced CIM scheme by creating a synergy between the CIM scheme  (introduced in \cite{raeisi2022cluster}) and the RIS technology. The proposed scheme is investigated in a realistic propagation environment in the presence of inter-cluster interference. At the same time, the available clusters are distributed randomly in both azimuth and elevation directions by adopting two-dimensional (2D) antenna arrays at all terminals to demonstrate the effectiveness of the proposed scheme in real channel conditions. In the proposed beam-gain optimized cluster selection CIM (BGCS-CIM) scheme, instead of indexing paths, clusters are indexed; consequently, by selecting the best effective path associated with each cluster and by considering the spatial diversity gain of other paths within the indexed cluster, the effect of inter-cluster interference can be reduced to construct an efficient CIM codebook. It is also shown that the entire cascaded channel is effective in achieving the optimal codebook even though one of the cascaded channels is completely shared between all the indexed entities. Extensive computer simulations reveal that the BGCS-CIM scheme outperforms the existing benchmarks. 
    
    \item We also integrate the CIM technique into an RIS-assisted transmission in mmWave massive MIMO systems to provide a virtual channel between the transmitter and receiver when the direct channel between them is blocked. For this purpose, we design a novel IM framework in which the clusters between the RIS and Rx get involved in the IM process. We also consider an outdoor scenario to make our analyses more practical.

    \item We consider a practical system model in which the channel between transmitter and RIS is line-of-sight (LOS), and the channel between the RIS and receiver is non-LOS (NLOS). In this proposed system model, we index clusters in the NLOS channel between the RIS and receiver to transmit extra information bits without adopting an extra RF chain. To the best of our knowledge, this system model for the RIS-assisted mmWave communication systems has not been studied before.
    
    \item We adopt a whitening transformation to perform decorrelation among the effective noises. Thanks to the whitening filter, the filtered noises corresponding to the different indexed clusters are independent; thus, we can constitute a joint probability density function (PDF) to calculate the conditional pairwise error probability (CPEP) in the case of erroneous cluster index detection. Therefore, we derive a closed-form expression for the average bit error rate (ABER) using the union upper bound. Our extensive computer simulations demonstrate the validity of the derived ABER expression by providing an upper bound to the simulation results. 
\end{itemize}

The remainder of this paper is organized as follows. Section \ref{sec: SystemModel} describes the system and channel model of a practical RIS-assisted mmWave massive MIMO. Then, ABER performance analysis is discussed in Section \ref{Sec: ABER} to obtain a closed-form expression for the proposed scheme. In Section \ref{Sec:Simulation Results}, we provide computer simulation results to validate the ABER performance of the proposed BGCS-CIM scheme and compare it with the available benchmarks. Finally, Section \ref{Sec: Conclusion} concludes this work. For improving the clarity, a list of frequently-used mathematical symbols and a list of acronyms are provided in Tables \ref{Tab:MathSymbols} and \ref{Tab:Acronyms}, respectively. 

\blfootnote{\textit{Notation}: In this paper, bold lower-case and bold upper-case letters represent column vectors and matrices, respectively. $(.)^H$, $(.)^T$, $|.|$, $||.||$, and $\diag(.)$ represent Hermitian, transpose, absolute value, norm, and diagonalization, respectively. $\mathbb{E}[.]$ and $\Re(.)$ denote expectation operator and real part of a complex number, respectively. $\mathcal{CN}(\mu,\sigma^2)$ indicates complex Gaussian distribution with mean $\mu$ and variance $\sigma^2$, while $\mathcal{N}(\mu,\sigma^2)$ expresses Gaussian distribution with mean $\mu$ and variance $\sigma^2$. $\mathcal{U}(a,b)$ depicts uniform distribution with parameters $a$ and $b$. $Q(x) = \frac{1}{\sqrt{2 \pi}}\int_x^{\infty} \exp{(-\frac{u^2}{2})} du$ is the $Q$-function, and $\Gamma(n) = (n-1)!$ represents the gamma function. $i$ is used as a local variable; hereupon, the values are assigned to $i$ in each equation are valid just for that equation.}

\begin{table}
\caption{List of acronyms.\label{Tab:Acronyms}}
\begin{tabular}{c| p{6cm} }
\toprule
IM & Index modulation \\ 
CIM & Cluster index modulation \\
SSM & Spatial scattering modulation\\
Tx & Transmitter\\
Rx & Receive\\
LOS & Line of sight\\
NLOS & Non-line of sight\\
RIS & Reconfigurable intelligence surface \\
BGCS & Beam-gain optimized cluster selection\\
InF & Inside factory\\
CSI & Channel state information\\
PEP & Pairwise error probability\\
CPEP & Conditional pairwise error probability\\
UPEP & Unconditional pairwise error probability\\
ABER & Average bit error rate\\
ZCA & Zero-phase components analysis\\
ULA & Uniform linear array\\
UPA & Uniform planar array\\
\bottomrule
\end{tabular}
\end{table}

\vspace{-5ex}

\section{System Model}\label{sec: SystemModel}
\vspace{-1ex}
Fig. \ref{fig:SystemModel} shows the proposed system model of the RIS-assisted mmWave transmission. Here, the transmitter (Tx) and receiver (Rx) are equipped with $N_t$ and $N_r$ antenna elements, respectively. In the mmWave band, due to the high penetration loss, blockages/outages occur frequently \cite{9685434}. Besides, it is shown that some corners and holes are challenging to cover in outdoor environments \cite{9955484}. Instead of implementing extra base stations, an RIS would be a cost-effective alternative to extend the communication range and cover such areas. Hence, in this paper, it is assumed that the direct channel between the Tx and Rx is completely blocked \cite{9955484, wang2020joint, guo2021joint, 9687593, 9685434, liu2020deep, 9103231, 9386246}. Thus, an RIS with $N$ reconfigurable reflecting elements is implemented to assist the mmWave communication system. Due to the rank-deficiency of the cascaded channels and inherent characteristics of the mmWave signals \cite{wang2020intelligent}, the Tx can only transmit a single stream. Hence, only one RF chain is available at the Tx, while $N_{RF}^r \ll N_r$ RF chains are available at the Rx. Thanks to the IM technique applied through clusters, the RIS can transmit an extra stream, i.e., the spatial domain information represented by $x_0$, without using an extra RF chain \cite{wang2019towards}\footnote{The proposed system model can be generalized for adopting $K$ RISs; consequently, the transmitter and the RIS can transmit $2K$ streams by adopting only $K$ RF chains.}. Since the RF chain is a costly and energy-consuming component in a communication system, decreasing the number of RF chains saves tremendous cost and energy. $\mathbf{G} \in \mathbb{C}^{N \times N_t}$ is the channel between the Tx and RIS, and $\mathbf{R} \in \mathbb{C}^{N_r \times N}$ is the channel between the RIS and Rx. In \cite{9386246} and \cite{yildirim2020modeling}, it is observed that the optimal placement for the RIS is the proximity of the Tx or Rx. In this paper, it is assumed that RIS is located close to the Tx. 
Since the Tx has a fixed location, the RIS can be positioned in the sight of Tx; accordingly, there is a robust LOS link between the Tx and the RIS \cite{wang2020intelligent, akdeniz2014millimeter}. On the other hand, the Rx is located far from the Tx; hence, due to the vulnerability of mmWave to blockage \cite{wang2020intelligent}, we cannot guarantee a LOS link between the RIS and Rx. Therefore, assuming a NLOS dominant channel between the RIS and Rx is more reasonable\footnote{As illustrated in \cite{akdeniz2014millimeter}, in a dense urban area when the carrier frequency is $28$ GHz, for distances less than $25$ m, the channel is almost LOS-dominant, while for distances more than $100$ m, the channel is NLOS-dominant with probability more than $80\%$.}\cite{akdeniz2014millimeter}.

\begin{figure}
    \centering
    \includegraphics[width = \columnwidth]{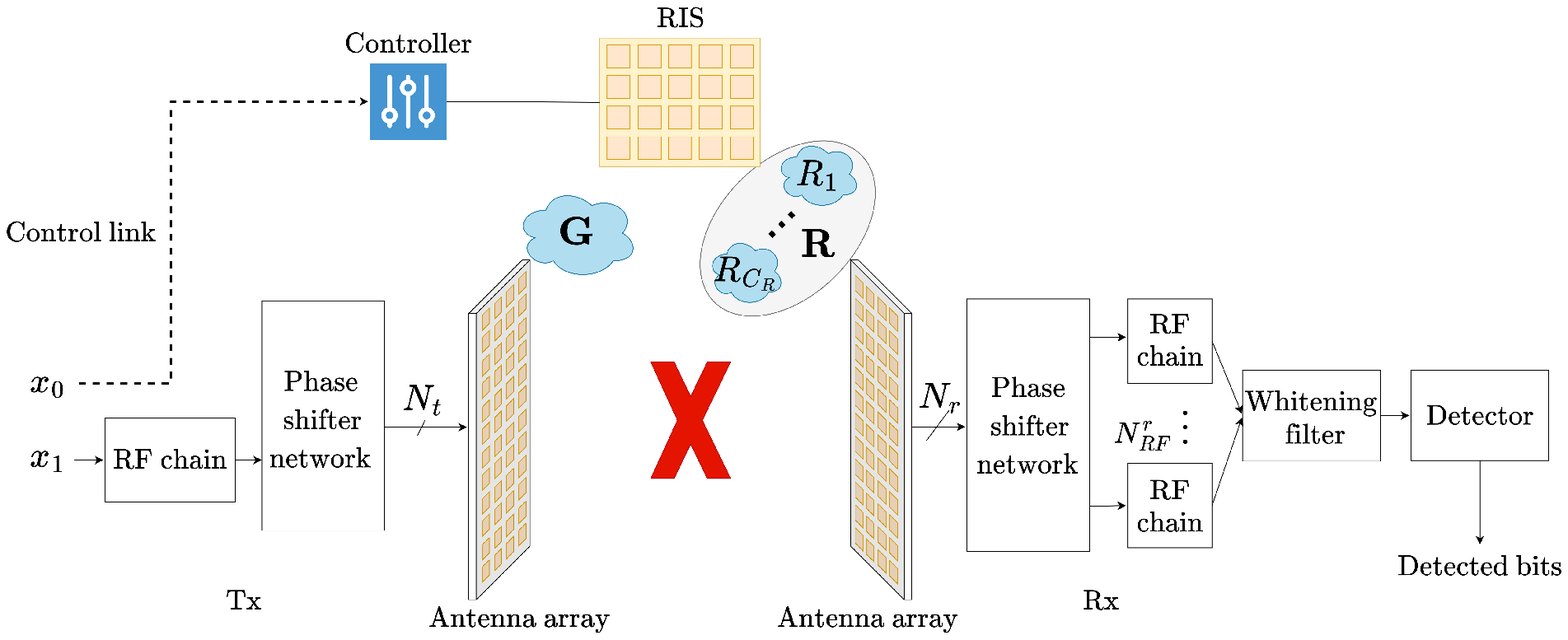}
    \caption{System Model.}
    \label{fig:SystemModel}
\end{figure}

\vspace{-3ex}

\subsection{Channel model}
\vspace{-1ex}
In this paper, we adopt the well-known Saleh-Valenzuela (S-V) channel model for mmWave communication systems as follows \cite{wang2020joint, wang2020intelligent, ying2020gmd, guo2021joint}:
\vspace{-1ex}
\begin{equation}\label{G channel model}
    \mathbf{G} = 
    \alpha_0 \mathbf{a}_{r,{ris}}(\phi_0^{r},\theta_0^{r}) \mathbf{a}_t^H(\varphi_0^t,\vartheta_0^t),
\end{equation}
\begin{equation}\label{R channel model}
    \mathbf{R} = 
    \sqrt{\frac{N N_r}{C_R L_R}} \sum_{c = 1}^{C_R} \sum_{l = 1}^{L_R} \beta_{c,l} \mathbf{a}_r(\varphi_{c,l}^r,\vartheta_{c,l}^r) \mathbf{a}_{t,{ris}}^H(\phi_{c,l}^t,\theta_{c,l}^t),
    \vspace{-0.5ex}
\end{equation}
where $C_R$ is the number of clusters, and $L_R$ is the number of paths in each cluster in the channel $\mathbf{R}$. The complex channel gain for the $l$th path within the $c$th cluster of channel $\mathbf{R}$ is represented by $\beta_{c,l}$ $\sim \mathcal{C}\mathcal{N}(0,10^{-0.1 PL(d)})$, wherein $PL(d)$ is the corresponding path loss dependent to the distance between two associated terminals, i.e., $d$, $\varphi_{c,l}^t$($\varphi_{c,l}^r)$ and $\vartheta_{c,l}^t$($\vartheta_{c,l}^r$) are the azimuth and elevation angles of departure (arrival) of the Tx (Rx), and $\phi_{c,l}^t$($\phi_{c,l}^r$) and $\theta_{c,l}^t$($\theta_{c,l}^r$) are the azimuth and elevation angles of departure (arrival) associated with the RIS.
$\alpha_0$ stands for the complex channel gain of the LOS path of channel $\mathbf{G}$. Moreover, $\mathbf{a}_t(.)$ and $\mathbf{a}_r(.)$ represent the array response vectors associated with the Tx and Rx, respectively; $\mathbf{a}_{t,{ris}}(.)$ and $\mathbf{a}_{r,{ris}}(.)$ show the transmit and receive array response vector associated with the RIS, respectively. In this paper, we adopt uniform planar array (UPA) at all terminals (i.e., Tx, Rx, and RIS); therefore, the normalized antenna array response in each terminal with azimuth angle $\omega^{az}$ and elevation angle $\omega^{el}$ is expressed as follows \cite{liu2020machine, mahmood20223}:
\vspace{-1ex}
\begin{equation}
    \mathbf{a}(\omega^{az}, \omega^{el}) = \frac{1}{\sqrt{N_a}}[1, e^{j\mathbf{k}^T \mathbf{p}_1}, e^{j\mathbf{k}^T \mathbf{p}_2} ..., e^{j\mathbf{k}^T \mathbf{p}_{N_a-1}}]^T,
\end{equation}
where $N_a = N_x \times N_y$ is the total number of antenna elements/reflectors; $N_x$ is the number of antenna elements/reflectors in $x$-axis, and $N_y$ is the number of antenna elements/reflectors in $y$-axis; $\mathbf{p}_n$ specifies the location of $n$th antenna element/reflector, and $\mathbf{k}$ is the wave-number defined as \cite{liu2020machine}:
\vspace{-1ex}
\begin{equation}
    \mathbf{k} = \frac{2 \pi}{\lambda}
    \begin{bmatrix}

        \sin (\omega^{el}) \cos (\omega^{az}),\
        \sin (\omega^{el}) \sin (\omega^{az}) \\
   \end{bmatrix}^T,
\end{equation}
\vspace{-2ex}
\begin{equation}
    \mathbf{p}_{n_x,n_y} = 
    \begin{bmatrix}
        n_x d_x,\
        n_y d_y\\
    \end{bmatrix}^T,
    \vspace{-1ex}
\end{equation}
where $\lambda$ is the wavelength, $0 \leq n_x \leq N_x-1$, and $0 \leq n_y \leq N_y-1$ denote the location of antenna elements/reflectors; $d_x$ and $d_y$ are the separation between two adjacent antenna elements/reflectors along $x$ and $y$ direction, respectively.

\vspace{-3ex}

\subsection{Signal model}
\vspace{-1ex}
As illustrated in Fig. \ref{fig:SystemModel}, the incoming bits constitute two information streams; $x_0$ is CIM symbol, while $x_1$ is modulated using $M$-ary constellation.
Accordingly, the total spectral efficiency is $\eta = \log_2M + \log_2B$ bits per channel use (bpcu), where $B$ is the CIM codebook order. The $M$-ary constellation symbol $s$ is transmitted to the RIS through the wireless channel $\mathbf{G}$. 
Besides, the Tx maps each CIM symbol into a phase vector from a pre-defined CIM codebook, i.e., $\mathcal{B} \in [ \mathbf{b}_1, \dots, \mathbf{b}_B ] \in \mathbb{C}^{N \times B}$, which is defined in Section \ref{Sec:CIM codebook}. Each phase vector $\mathbf{b}_i$ represents a specific direction associated with a specific cluster. Afterward, based on the CIM bits and with the help of a smart controller \cite{wang2020joint,wang2020intelligent,ying2020gmd}, the corresponding phase vector will be sent to the RIS to adjust the phases of the received signal at each element of the RIS \cite{gopi2020intelligent}. In other words, the RIS is exploited as a beamformer to reflect the signal passively over a specific cluster in the mmWave channel \cite{gopi2020intelligent}. Denoting $\psi_i$ as the phase shift associated with the $i$th passive element of the RIS, the RIS phase shift matrix associated with vector $\mathbf{b}$, i.e.,  $\mathbf{\Psi}_{\mathbf{b}} \in \mathbb{C}^{N \times N}$, is defined as:
\vspace{-1ex}
\begin{equation}
\vspace{-1ex}
    \begin{split}
        \mathbf{\Psi}_{\mathbf{b}} & \triangleq \diag (e^{j\psi_1}, \dots, e^{j\psi_{N}}) = \diag(\mathbf{b}),
    \end{split}
\end{equation}
where $\mathbf{b} \in \mathcal{B}$ is the codeword corresponding to the transmitted CIM symbol. For the notation simplicity, we represent the effective channel associated with vector $\mathbf{b}$ as $\mathbf{H_{eff}^{\mathbf{b}}} = G_t G_r\mathbf{R}\mathbf{\Psi}_{\mathbf{b}}\mathbf{G}$, where $G_t$ and $G_r$ are the transmitter and receiver antenna gains, respectively.
The received signal at the Rx can be expressed as:
\vspace{-1ex}
\begin{equation}
    \mathbf{y} = \sqrt{P} \mathbf{H_{eff}^{\mathbf{b}}}\mathbf{f}_t s + \mathbf{n},
    \vspace{-2ex}
\end{equation}
where $\mathbf{n} \in \mathbb{C}^{N_r \times 1} \sim \mathcal{CN}(0,\sigma^2)$ is the additive noise component at the Rx, $\mathbf{f}_t \in \mathbb{C}^{N_t \times 1}$ is the analog beamformer at the Tx, and $P$ is the transmit power. The Tx transmits the signal through the LOS path due to its dominance in channel $\mathbf{G}$; thus, we can consider $\mathbf{f}_t = \mathbf{a}_t(\varphi_0^t, \vartheta_0^t)$. In this paper, we assume that the perfect channel state information (CSI) is available \cite{wang2020joint, wang2020intelligent, guo2021joint}. More specifically, channel estimation for the RIS-assisted mmWave systems is discussed in \cite{9685434}\footnote{Authors in \cite{9685434} proposed a low-complexity channel estimation using semi-passive RIS. In this algorithm, few simplified receiver units with only 1-bit quantization at the RIS are randomly connected to a small fraction of reflectors; therefore, the RIS actively gets involved in the channel estimation procedure.}, \cite{9103231, liu2020deep}. The received signal $\mathbf{z} \in \mathbb{C}^{B \times 1}$ after applying analog combiner is given as:
\vspace{-1ex}
\begin{equation}\label{Received Signal after combiner}
\vspace{-1ex}
    \mathbf{z} = \sqrt{P} \mathbf{W}^H \mathbf{H_{eff}^{\mathbf{b}}} \mathbf{f}_t s + \mathbf{W}^H \mathbf{n},
\end{equation}
where $\mathbf{W} = [\mathbf{w}_1, \dots, \mathbf{w}_B] \in \mathbb{C}^{N_r \times B}$ is the analog combiner matrix wherein $\mathbf{w}_i \in \mathbb{C}^{N_r \times 1}$ is the analog combiner vector associated with the $i$th indexed cluster. Section \ref{Sec:CIM codebook} discusses the method to determine $\mathbf{W}$.

It is important to note that even though $\mathbf{n}$ is a white noise vector, the effective noise $\mathbf{W}^H \mathbf{n}$ is colored; consequently, the effective noises of different clusters are correlated. To derive a closed-form expression for CPEP in Section \ref{Sec: ABER}, we need to whiten the effective noise through a whitening filter\footnote{It is worth mentioning that the designed pre-whitening filter whitens the effective noises for theoretical derivations. Accordingly, the whitening process does not get involved in the codebook construction procedure.}. 
\vspace{-1ex}
\begin{lemma}\label{lemm: Whitening Filter}
The whitening filter for the proposed system can be obtained as follows: 
\vspace{-1ex}
\begin{equation}
    \mathbf{B} = \mathbf{P}^{-\frac{1}{2}} \mathbf{V}^{-\frac{1}{2}},
    \vspace{-2ex}
\end{equation}
where $\mathbf{P}$ and $\mathbf{V}$ are the correlation matrix and diagonal variance matrix corresponding to the effective noise $\mathbf{W}^H \mathbf{n}$, respectively.
\end{lemma}
\vspace{-2ex}
\begin{proof}
    Please see Appendix \ref{Appndx: Proof of whitening filter}.
\end{proof}
\vspace{-1ex}
Hence, by applying whitening filter to the received signal $\mathbf{z}$, the filtered signal $\mathbf{z_F} \in \mathbb{C}^{B \times 1}$ is as:
\begin{equation}\label{eq:Filtered signal}
    \mathbf{z_F} = \sqrt{P}  \mathbf{F}_r^H \mathbf{H_{eff}^{\mathbf{b}}} \mathbf{f}_t s + \mathbf{F}_r^H \mathbf{n}.
\end{equation}
where $\mathbf{F}_r = \mathbf{W}\mathbf{B}^H = [\mathbf{f}_{r,1}, \dots, \mathbf{f}_{r,B}]$.
We adopt maximum-likelihood (ML) detector to resolve the transmitted symbols from the received filtered signal as follows\footnote{It is worth mentioning that although the ML detector is the best option in terms of BER-vs-transmitted power performance, it leads to a high-complexity receiver. In order to decrease the detector’s complexity, sub-optimal detectors should be designed by considering different use cases of the CIM scheme, which is beyond the scope of this paper and is left for future works.}:
\begin{equation}
    [\hat{c},\hat{s}] = \arg \min_{c,s} \big | \mathbf{z_F}(c) - \sqrt{P} \mathbf{f}_{r,c}^H \mathbf{H_{eff}^{\mathbf{b}}} \mathbf{f}_t s \big |^2,
\end{equation}
where $\hat{c}$ and $\hat{s}$ are the estimated CIM and $M$-ary symbols, respectively; and $\mathbf{z_F}(c)$ is the $c$th element in $\mathbf{z_F}$.

\vspace{-2ex}
\subsection{Constructing CIM codebook and analog combiner}\label{Sec:CIM codebook}
\vspace{-1ex}
This subsection presents a scheme for constructing CIM codebook, i.e., $\mathcal{B}$, and analog combiner, i.e., $\mathbf{W}$, for the RIS-assisted mmWave communication systems.
Let us define the effective path gain of the $v$th path in the $u$th cluster as follows:
{\fontsize{9.7}{10}
\begin{equation}\label{eq:EffectivePath}
    \mathcal{G}_{u,v} =
    \frac{\sqrt{P}G_tG_r}{\sqrt{N}}\mathbf{a}_r^H(\varphi_{u,v}^r, \vartheta_{u,v}^r) \mathbf{R} \diag(\mathbf{a}_{t,{ris}}(\phi_{u,v}^t, \theta_{u,v}^t))\mathbf{G}\mathbf{f}_t.
    \vspace{-0.5ex}
\end{equation}}

We assume each cluster is represented by the best effective path as follows:
\vspace{-2ex}
\begin{equation}\label{eq:best path of each cluster}
\vspace{-1ex}
    p_{c} = \arg \max_{i = 1, \dots, L_R} |\mathcal{G}_{c,i}|^2,
\end{equation}
where $p_{c}, c = 1, \dots, C_R$, is the best path index in the $c$th cluster.
Henceforth, we show the set of all clusters with their representative path as $\mathcal{P} = \{p_1, \dots, p_{C_R}\}$.




Thereupon, the BGCS-CIM scheme constructs the CIM-codebook consecutively by adopting the following criteria:
\vspace{-1ex}
\begin{equation}\label{eq:BGCS}
\vspace{-1ex}
    \begin{split}
        b =
        \arg \ \max_{i =  1 , \dots, C_R} |\mathcal{G}_{i,p_i}|^2.
    \end{split}
\end{equation}
Notice that $b$ gives the best choice among all possible options of the corresponding scheme. To construct the CIM codebook, we need to choose $B$ out of $C_R$ clusters for indexing. It is worth noting that the codewords carry the information to adjust RIS phase shifts. The Tx should select a codeword according to the CIM symbol and transfer it to the RIS through the control link. Finally, analog combiner vector associated with $i$th indexed cluster can be determined as $\mathbf{w}_i = \mathbf{a}_r(\varphi_{i,p_i}^r,\vartheta_{i,p_i}^r)$. Algorithm \ref{alg:CIM codebook construction} presents the CIM codebook and analog combiner construction procedure according to the BGCS-CIM scheme\footnote{The primary purpose of the BGCS-CIM algorithm is to index the clusters in an RIS-assisted mmWave system and reduce the interference among the indexed entities. The correlation among paths and clusters is an inseparable characteristic of the mmWave channel. Nevertheless, the essential point is exploiting them concurrently while decreasing the entailed interference. BGCS-CIM proposes an algorithm to reduce this interference in the analog stage without involving a baseband process.}. 

\begin{algorithm}[t]
\footnotesize
\caption{CIM codebook and analog combiner construction algorithm according to the BGCS-CIM scheme.}\label{alg:CIM codebook construction}
\KwIn{$\mathbf{R}, \mathbf{G}, \mathbf{f}_t, \phi, \theta, \varphi, \vartheta, \mathbf{a}(.)$}
\KwOut{$\mathcal{B}$}
        Calculate the effective path gain for all paths using (\ref{eq:EffectivePath}).\\
    Find the best effective path for each cluster using (\ref{eq:best path of each cluster}).\\
Construct CIM codeword index candidates:
$\mathcal{P} = \{ p_1, \dots, p_{C_R} \}.$\\
\For{$k \gets 1 \ to \ B$}{
    Calculate the best cluster $b_k$ among $\mathcal{P}$ by adopting BGCS-CIM scheme in (\ref{eq:BGCS}).\\
    calculate $k$th codeword using $\mathbf{b}_k = \mathbf{a}_{t,ris}(\phi_{b_k,p_{b_k}}^t,\theta_{b_k,p_{b_k}}^t)$.\\ 
    calculate $k$th combiner vector using $\mathbf{w}_k = \mathbf{a}_r(\varphi_{b_k,p_{b_k}}^r,\vartheta_{b_k,p_{b_k}}^r)$.\\
    Update $\mathcal{P}$:
    $\mathcal{P} = \mathcal{P} - \{b_k\}.$
}
Construct CIM codebook: $\mathcal{B} = \{ \mathbf{b}_{1}, \dots, \mathbf{b}_B \}.$\\
Construct analog combiner: $\mathbf{W} = [\mathbf{w}_{1}, \dots, \mathbf{w}_{B}]$.
\end{algorithm}\normalsize

\vspace{-3ex}

\section{ABER Performance Analysis} \label{Sec: ABER}
\vspace{-0.5ex}
In this section we derive a closed-form expression for the upper bound on ABER to validate performance of the proposed system. The upper bound on ABER can be stated as:
\vspace{-1ex}
\begin{equation}
\vspace{-1ex}
    \textrm{ABER} \leq \frac{1}{\eta M B} \sum_{c^*,s^*} \sum_{\hat{c}, \hat{s}} E_b(\{ c^*, s^* \} \rightarrow \{ \hat{c}, \hat{s} \}) \mathrm{UPEP},
\end{equation}
where $E_b(\{c^*, s^*\} \rightarrow \{ \hat{c}, \hat{s} \})$ is the total number of erroneous bits of both $M$-ary and CIM symbols, and $\mathrm{UPEP}$ is the unconditional pairwise error probability. To obtain the upper bound, we first state the CPEP as follows:
\vspace{-1ex}
\begin{equation}\label{CPEP_initial}
\vspace{-1ex}
    \begin{split}
        \mathbb{P} (\{ & c^*,s^* \} \rightarrow \{ \hat{c}, \hat{s} \} | \alpha_0, \beta_{1,1}, \dots, \beta_{C_R, L_R}) = \mathbb{P}(|\mathbf{z_F}(c^*) \\
        & - \sqrt{P} \mathbf{f}_{r,c^*}^H \mathbf{H_{eff}^{\mathbf{b^*}}} \mathbf{f}_t s^*|^2 
        > |\mathbf{z_F}(\hat{c}) - \sqrt{P} \mathbf{f}_{r,\hat{c}}^H \mathbf{H_{eff}^{\mathbf{\hat{b}}}} \mathbf{f}_t \hat{s}|^2),
    \end{split}
\end{equation}
where $(.)^*$ and $\hat{(.)}$ depict the transmitted and detected symbols, respectively. By substituting $\mathbf{z_F}(c^*)$ from (\ref{eq:Filtered signal}) into (\ref{CPEP_initial}), we can rewrite the CPEP expression as:
\vspace{-2ex}
\begin{equation}\label{CPEP}
    \begin{split}
        \mathbb{P} (\{ c^*,s^* \} & \rightarrow  \{ \hat{c}, \hat{s} \} | \alpha_0, \beta_{1,1}, \dots, \beta_{C_R, L_R}) \\
        & = \mathbb{P}(|\mathbf{f}_{r,c^*}^H \mathbf{n}|^2 > |\mathbf{z_F}(\hat{c}) - \sqrt{P} \mathbf{f}_{r,\hat{c}}^H \mathbf{H_{eff}^{\mathbf{\hat{b}}}}\mathbf{f}_t \hat{s}|^2).
    \end{split}
    \vspace{-2ex}
\end{equation}
In the following, we consider two cases of correct cluster index detection, i.e., $\hat{c} = c^*$, and erroneous cluster index detection, i.e., $\hat{c} \neq c^*$. 

\vspace{-3ex}
\subsection{Correct Cluster Index Detection $\hat{c} = c^*$}
\vspace{-1ex}
In the case of correct cluster index detection, we assume that the detector resolves the CIM symbol correctly; hence, the error comes from the $M$-ary symbol. To obtain the CPEP expression for this case, (\ref{CPEP}) can be updated as:
\begin{equation} \label{eq: CPEP correct CIM detection}
    \begin{split}
        & \mathbb{P} (\{ c^*,s^* \} \rightarrow \{ c^*, \hat{s} \} | \alpha_0, \beta_{1,1}, \dots, \beta_{C_R, L_R}) \\
        & = \mathbb{P}(|\mathbf{f}_{r,c^*}^H \mathbf{n}|^2 > |\mathbf{f}_{r,c^*}^H \mathbf{n} + \sqrt{P} \mathbf{f}_{r,c^*}^H \mathbf{H_{eff}^{\mathbf{b^*}}} \mathbf{f}_t (s^* - \hat{s})|^2).
    \end{split}
\end{equation}

\begin{lemma} \label{PEP for correct cluster index detection}
In the case of correct cluster index detection, the analytical expression for $\mathrm{UPEP}$ is as follows:
\begin{equation}\label{eq:correct CIM UPEP}
\begin{split}
    \mathbb{P} (\{ c^*, & s^* \}\rightarrow \{ c^*, \hat{s} \}) \\
    & = \mathbb{E} \bigg[ Q \Big( \frac{| \mathbf{f}_{r,c^*}^H \mathbf{H_{eff}^{\mathbf{b^*}}} \mathbf{f}_t (s^* - \hat{s}) |^2}{\sqrt{2} \sigma || \mathbf{f}_{r,c^*} \mathbf{f}_{r,c^*}^H \mathbf{H_{eff}^{\mathbf{b^*}}} \mathbf{f}_t (s^* - \hat{s}) ||} \Big) \bigg].
\end{split}
\end{equation}
\end{lemma}

\begin{proof}
    Please see Appendix \ref{Proof of PEP for correct cluster index detection}.
\end{proof}

\vspace{-5ex}

\subsection{Erroneous Cluster Index Detection $\hat{c} \neq c^*$}
When the detector resolves the CIM symbol erroneously, error can propagate to the $M$-ary symbol. In this case, (\ref{CPEP}) can be rewritten as follows:
\begin{equation}\label{eq:Erroneous CIM UPEP}
\begin{split}
    \mathbb{P} (\{ c^*, s^* \} & \rightarrow \{ \hat{c}, \hat{s} \} | \alpha_0, \beta_{1,1}, \dots, \beta_{C_R, L_R}) = \mathbb{P} (|\mathbf{f}_{r,c^*}^H \mathbf{n}|^2\\ 
    & - |\mathbf{f}_{r,\hat{c}}^H \mathbf{n}
    + \sqrt{P} \mathbf{f}_{r,\hat{c}}^H  (\mathbf{H_{eff}^{b^*}} s^* - \mathbf{H_{eff}^{\hat{b}}} \hat{s} ) \mathbf{f}_t|^2 > 0).
\end{split}
\end{equation}

\begin{lemma}\label{PEP for erroneous cluster index detection}
In the case of erroneous cluster index detection, the analytical expression for $\mathrm{UPEP}$ can be obtained as:
\begin{equation}\label{UPEP for errouneous index detection}
\begin{split}
    \mathbb{P} & (\{ c^*,s^* \} \rightarrow \{ \hat{c}, \hat{s} \}) \\
    & = \mathbb{E} \bigg[ \frac{1}{2} \exp{\Big( - \frac{|\sqrt{P} \mathbf{f}_{r,\hat{c}}^H (\mathbf{H_{eff}^{b^*}} s^* - \mathbf{H_{eff}^{\hat{b}}} \hat{s} ) \mathbf{f}_t |^2}{2 \sigma^2} \Big)} \bigg].
\end{split}
\end{equation}
\end{lemma}

\begin{proof}
    Please see Appendix \ref{Proof of PEP for erroneous cluster index detection}.
\end{proof}

Using (\ref{eq:correct CIM UPEP}) and (\ref{UPEP for errouneous index detection}), we can express $\mathrm{UPEP}$ as:

\vspace{-2ex}
\begin{equation}
    \begin{split}
        &\mathrm{UPEP} = \\
        &
        {\fontsize{9.4}{9}
        \begin{cases}
            \mathbb{E} \bigg[ Q \Big( \frac{| \mathbf{f}_{r,c^*}^H \mathbf{H_{eff}^{\mathbf{b^*}}} \mathbf{f}_t (s^* - \hat{s}) |^2}{\sqrt{2} \sigma || \mathbf{f}_{r,c^*} \mathbf{f}_{r,c^*}^H \mathbf{H_{eff}^{\mathbf{b^*}}} \mathbf{f}_t (s^* - \hat{s}) ||} \Big) \bigg], \textrm{if} \  \{ c^*, s^* \} \rightarrow \{ c^*, \hat{s} \}, \\
            \\
            \mathbb{E} \bigg[ \frac{1}{2} \exp{\Big( - \frac{|\sqrt{P} \mathbf{f}_{r,\hat{c}}^H (\mathbf{H_{eff}^{b^*}} s^* - \mathbf{H_{eff}^{\hat{b}}} \hat{s} ) \mathbf{f}_t |^2}{2 \sigma^2} \Big)} \bigg],  \textrm{if} \  \{ c^*, s^* \} \rightarrow \{ \hat{c}, \hat{s} \}.
        \end{cases}
        }
    \end{split}
\end{equation}

\vspace{-4ex}

\section{Simulation Results}\label{Sec:Simulation Results}
\vspace{-1ex}
This section provides computer simulation results for the proposed BGCS-CIM scheme for the RIS-assisted mmWave system. We first describe the simulation setup for a practical scenario in an outdoor environment. Then, we validate our simulation results with the derived analytical upper bound. We also compare the proposed scheme with the simple CIM scheme introduced in \cite{raeisi2022cluster}, the SSM-based scheme proposed in \cite{jiang2021generalized}, the orthogonal match pursuit (OMP) method proposed in \cite{el2014spatially}, and the random cluster selection (RCS) scheme in which a cluster is randomly selected for transmitting the conventional $M$-ary signaling. Furthermore, we investigate the proposed system performance for different antenna elements/reflectors and different channel sparsity. Ultimately, we examine the robustness of the proposed scheme with respect to the angular information change.

\begin{figure}
    \centering
    \includegraphics[scale = 0.15]{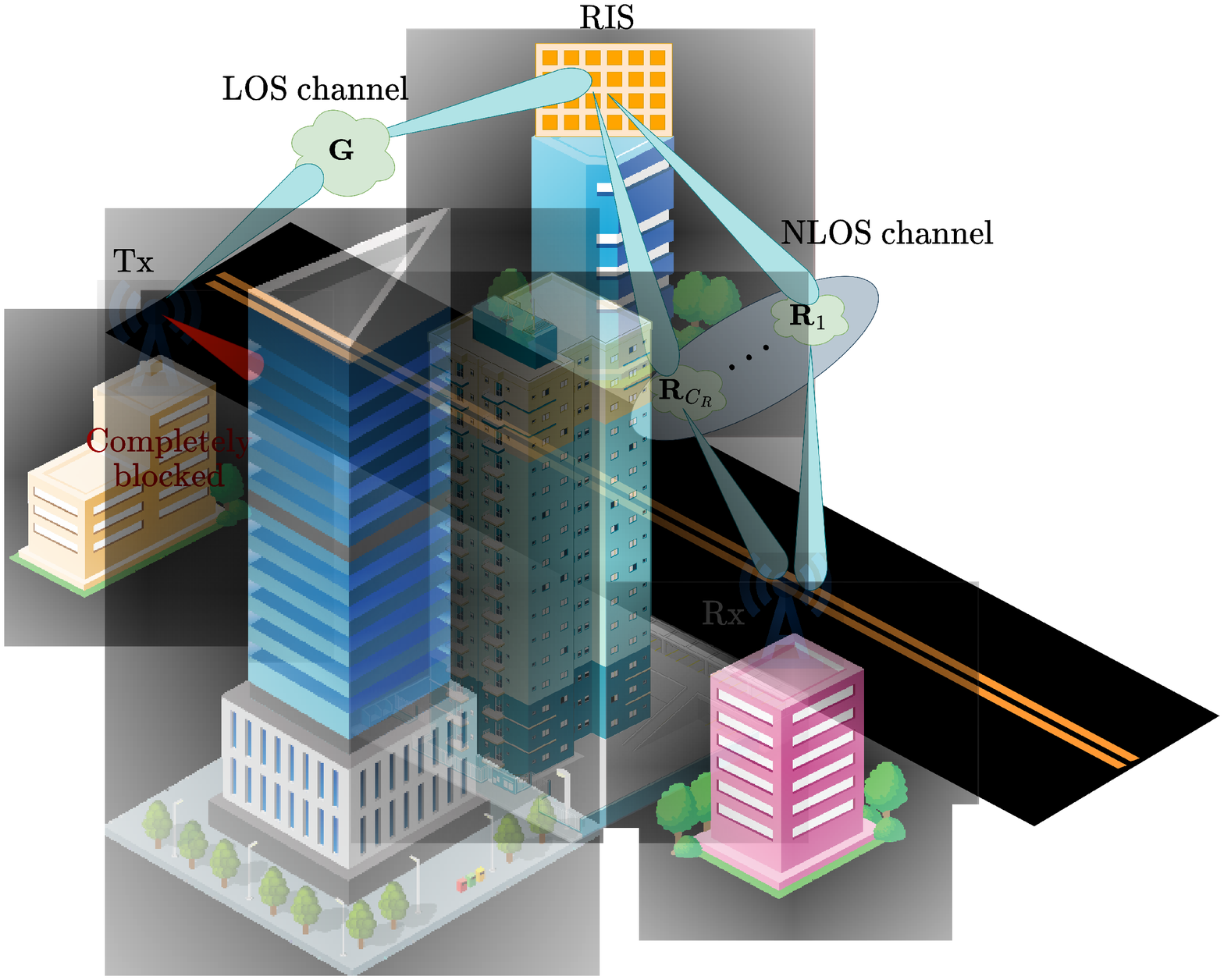}
    \caption{Considered outdoor scenario with RIS deployment on top of the roof of a neighboring building.}
    \label{fig:considered scenario}
\end{figure}

\begin{table}
  \centering
  \caption{Simulation parameters.}
  \vspace{-3ex}
  \begin{tabular}{c}
  \begin{tabular}[t]{ c  c  c }
  \toprule
  \multicolumn{2}{c}{\textbf{Parameter}}    & \textbf{Assumed value}  \\
  \toprule
  \multicolumn{2}{c}{Tx location}      & $(3 \ \textrm{m}, 0, 12 \ \textrm{m})$   \\ [1 pt]
  \hline
  \multicolumn{2}{c}{RIS location}     & $(0, 3 \ \textrm{m}, 15 \ \textrm{m})$    \\ [1 pt]
  \hline
  \multicolumn{2}{c}{Rx location}     & $(3 \ \textrm{m}, 103 \ \textrm{m}, 6 \ \textrm{m})$    \\ [1 pt]
  \hline
  \multicolumn{2}{c}{PSD}     & $-174$ dBm    \\
  \hline
  \multicolumn{2}{c}{$G_t = G_r$}      & $24.5$ dBi   \\ 
  \hline
  \multirow{3}*{LOS}      & $a$ & $61.4$   \\
       & $b$ & $2$  \\
       & $\sigma_{\xi}$ & $5.8$ \\  [1 pt]
  \hline 
  \multirow{3}*{NLOS}      & $a$ & $72$   \\
       & $b$ & $2.92$  \\
       & $\sigma_{\xi}$ & $8.7$ \\ 
  \bottomrule
\end{tabular}
\begin{tabular}[t]{c c}
    \toprule
    \textbf{Parameter} & \textbf{Assumed value} \\
    \toprule
    $f$ & $28$ GHz  \\
    \hline
    $BW$ & $100$ MHz\\
    \hline
    $L_R$ & $10$ \\
    \hline
    $C_R$ & $8$ \\
    \hline
    $N$ & $10 \times 10$\\
    \hline
    $N_t = N_r$ & $8 \times 8$ \\
    \hline
    $d_x = d_y$ & $\lambda / 2$\\
    \hline
    $\phi_m^t, \varphi_m^t$ & $\mathcal{U}[0,2\pi)$\\
    \hline
    $\theta_m^t, \vartheta_m^t$ & $\mathcal{U}[0,\pi)$\\
    \hline
    $\sigma_{\phi_m^t}, \sigma_{\theta_m^t}$ & $7.5^\circ$\\
    \hline
    $\sigma_{\varphi_m^r}, \sigma_{\vartheta_m^r}$ & $7.5^\circ$\\
    \bottomrule 
\end{tabular}
\end{tabular}
  \label{tab:System parameters}
\end{table}

\begin{figure}
    \centering
        \includegraphics[width = \columnwidth]{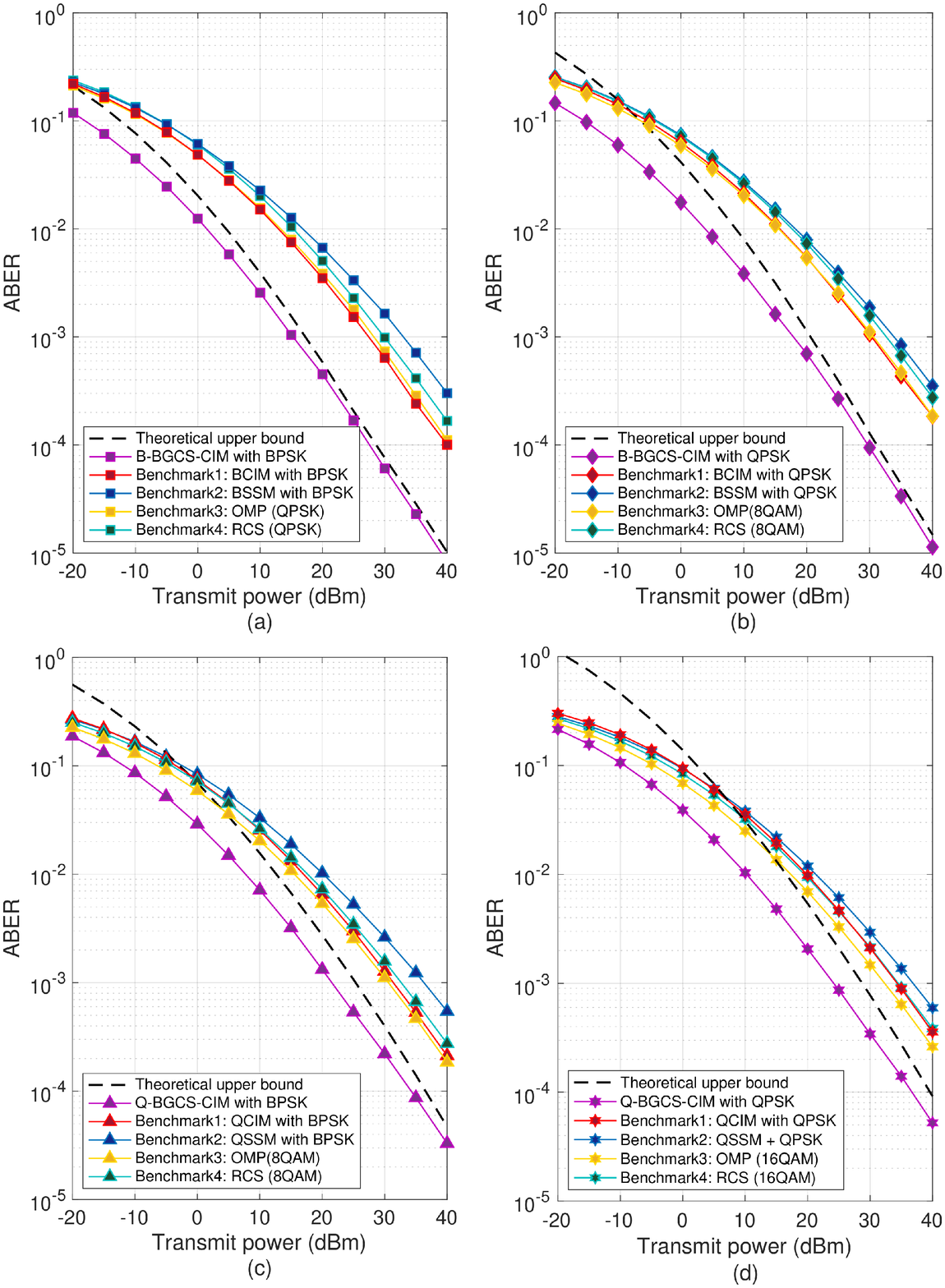}
    \caption{ABER performance, analytical validation and benchmarks comparison: (a) B-IM with BPSK, (b) B-IM with QPSK, (c) Q-IM with BPSK, (d) Q-IM with QPSK.}
    \label{fig:CIMvsSSM}
\end{figure}

\vspace{-4ex}
\subsection{Simulation setup}
\vspace{-1ex}
In this subsection, we investigate a communication scenario between two BSs (back-haul scenario) in an urban area \cite{akdeniz2014millimeter}. The coordinates of the Tx, RIS, and Rx are $(3 \ \textrm{m}, 0, 12 \ \textrm{m})$, $(0, 3 \ \textrm{m}, 15 \ \textrm{m})$, and $(3 \ \textrm{m}, 103 \ \textrm{m}, 6 \ \textrm{m})$ in 3D Cartesian coordinate system, respectively. It is worth emphasizing that the channel between Tx and Rx is completely blocked and communication is possible only through a cascaded channel via an RIS. As depicted in \cite{akdeniz2014millimeter}, in a dense urban area and for a carrier frequency of $28$ GHz, the channel is almost LOS-dominant when the distance is less than $25$ m. On the other hand, if the distance between transmit and receive terminal is more than $100$ m, it is shown that the channel is NLOS-dominant. Therefore, in this paper, we assume that the Tx-RIS channel, which is a short-range link, is LOS-dominant, while the long-range link between the RIS and Rx is NLOS-dominant. Fig. \ref{fig:considered scenario} depicts a scheme of considered scenario.
The path loss model is as follows \cite{akdeniz2014millimeter}:
\vspace{-1ex}
\begin{equation}\label{eq:pathloss model}
    PL(d) \ [dB] = a + 10 b \log_{10} (d) + \xi,
    \vspace{-1ex}
\end{equation}
where $\xi\sim \mathcal{N}(0, \sigma_{\xi}^2)$, and $d$ is the distance between the Tx and RIS (or the RIS and Rx). The parameters $a$, $b$, and $\sigma_{\xi}$ are given according to Table \ref{tab:System parameters} which are collected from experimental measurements at $28$ GHz carrier frequency\footnote{In (\ref{eq:pathloss model}), $\xi$ is a random variable that describes large-scale shadow fading with a zero mean Gaussian distribution and standard deviation $\sigma_{\xi}$.} \cite{akdeniz2014millimeter}. In this paper, the noise power spectral density (PSD) is assumed to be $-174$ dBm/Hz \cite{9390210}, and the bandwidth ($BW$) is supposed to be $100$ MHz \cite{etsitr138901}; hence, the noise power can be calculated as  $\sigma^2 = \textrm{PSD} + 10\log_{10}(BW) = -94$ dBm.

Regarding the channel $\mathbf{R}$, we assume that there are $C_R = 8$ clusters in the mmWave propagation environment; and in each cluster, there exists $L_R = 10$ paths \cite{liu2020machine, el2014spatially}. Each cluster has departure/arrival azimuth and elevation mean angles ($\phi_m^t, \theta_m^t, \varphi_m^r, \vartheta_m^r$), which are distributed uniformly in an interval of $[0,2\pi)$ and $[0,\pi)$, respectively; and the angular spread (standard deviation) of departure/arrival of both azimuth and elevation ($\sigma_{\phi_m^t}, \sigma_{\theta_m^t}, \sigma_{\varphi_m^r}, \sigma_{\vartheta_m^r}$) domains are supposed to be $7.5^\circ$. Within each cluster, the azimuth and elevation angles corresponding to each path follow Laplacian distribution \cite{liu2020machine, el2014spatially}. The transmitter and receiver antenna gains ($G_t$ and $G_r$) are assumed to be $24.5$ dBi \cite{wang2020joint, ning2021terahertz}. Table \ref{tab:System parameters} shows a summary of the assumed values of the system parameters. 

\vspace{-3ex}

\begin{figure}
    \centering
    \includegraphics[scale = 0.4]{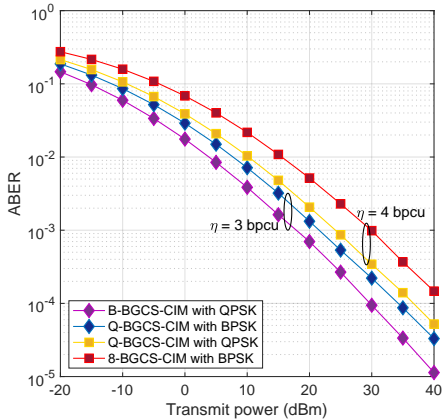}
    \caption{ABER performance for different CIM orders.}
    \label{fig:CIM_EuclideanDistance}
\end{figure}

\subsection{Analytical validation and benchmark comparison}
\vspace{-0.5ex}
In this subsection, we validate our simulations with the derived analytical upper bound and compare the proposed scheme with four benchmarks, i.e., simple CIM \cite{raeisi2022cluster}, SSM \cite{ding2017spatial, jiang2021generalized}, OMP \cite{el2014spatially}, and RCS, for the RIS-assisted mmWave systems. It is worth noting that the simple CIM is designed to index the clusters in a sparse channel between transmit and receive terminals. Since we need to index the clusters in the RIS-Rx channel, we apply the simple CIM in that channel. However, BGCS-CIM creates a synergy between the CIM scheme and the RIS technology. In the BGCS-CIM scheme, we consider the entire cascaded channel in the construction procedure of the codebook. 
Fig. \ref{fig:CIMvsSSM} illustrates the ABER performance versus the transmit power ($P$) for four different signaling schemes: (a) binary IM (B-IM) with binary phase-shift keying (BPSK), (b) B-IM with quadrature PSK (QPSK), (c) quadrature IM (Q-IM) with BPSK, and (d) Q-IM with quadrature PSK (QPSK). It is worth mentioning that, spatial information domain includes one bit and two bits for B-IM and Q-IM, respectively.
As shown in the Fig. \ref{fig:CIMvsSSM}, the analytical expression is an upper bound on the BGCS-CIM scheme, which validates the simulation results\footnote{The gap between the upper bound and simulation results increases with the increasing spectral efficiency. Since the upper bound is an approximation, with increasing the spectral efficiency, the accuracy of the approximation decreases. Nonetheless, for higher transmit powers, we will see that the theoretical upper bound approaches the simulation curve.}.
The simulation results reveal that BGCS-CIM scheme can outperform the benchmark schemes. Regarding Benchmark 1, i.e., simple CIM, at first glance, it seems that we should have a similar performance because the Tx-RIS LOS channel is shared among all the clusters in the RIS-Rx channel. However, simulation results reveal that the BGCS-CIM scheme can boost ABER performance significantly. These results prove the importance of the entire cascaded channel in constructing an effective CIM codebook. On the other hand, the superiority of CIM-based schemes, i.e., simple CIM and BGCS-CIM, over SSM is due to the inherent characteristic of CIM-based schemes in indexing clusters instead of paths. In other words, in the SSM scheme, under realistic circumstances of correlation among paths, two indexing paths can be near to each other, which results in power leakage in unfavorable paths; while by indexing clusters instead of paths, we can reduce the possibility of indexing highly correlated paths.

\begin{figure}
    \centering
    \includegraphics[scale = 0.35]{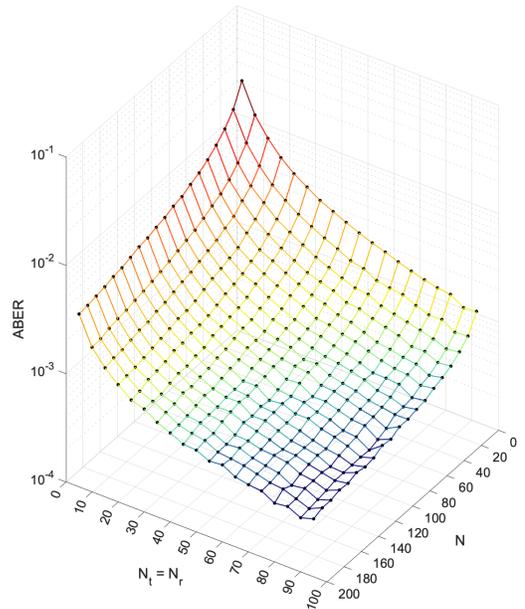}
    \caption{ABER for different numbers of antenna elments and RIS reflectors for B-BGCS-CIM with QPSK signaling.}
    \label{fig:ABER of different array sizes}
\end{figure}

\begin{figure*}
    \centering
    \includegraphics[scale = 0.27
    ]{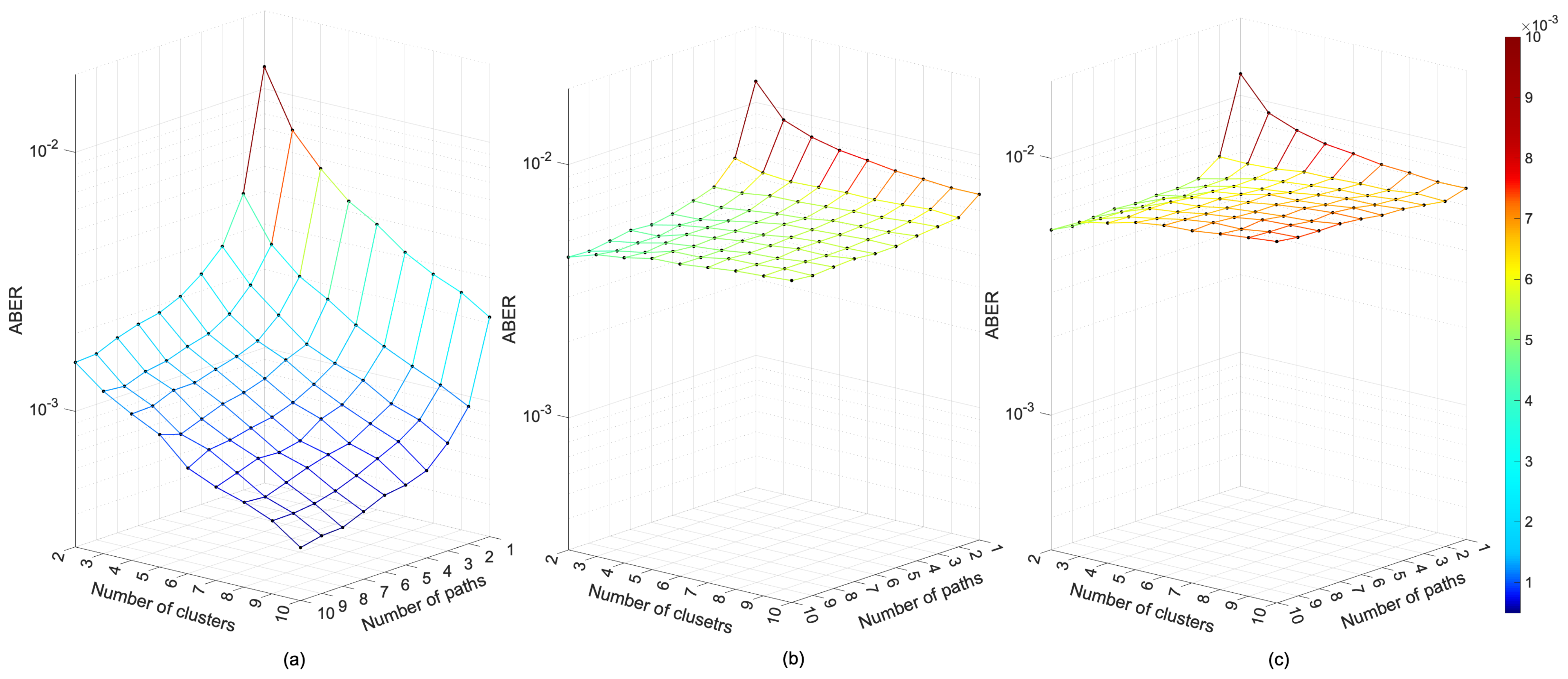}
    \caption{ABER performance for different IM schemes with respect to different channel sparsity: (a) B-BGCS-CIM with QPSK, (b) B-CIM with QPSK, (c) B-SSM with QPSK.}
    \label{fig:ABER_ChannelSparsity3D}
    \vspace{-3ex}
\end{figure*}

\begin{figure}
    \centering
    \includegraphics[scale = 0.43]{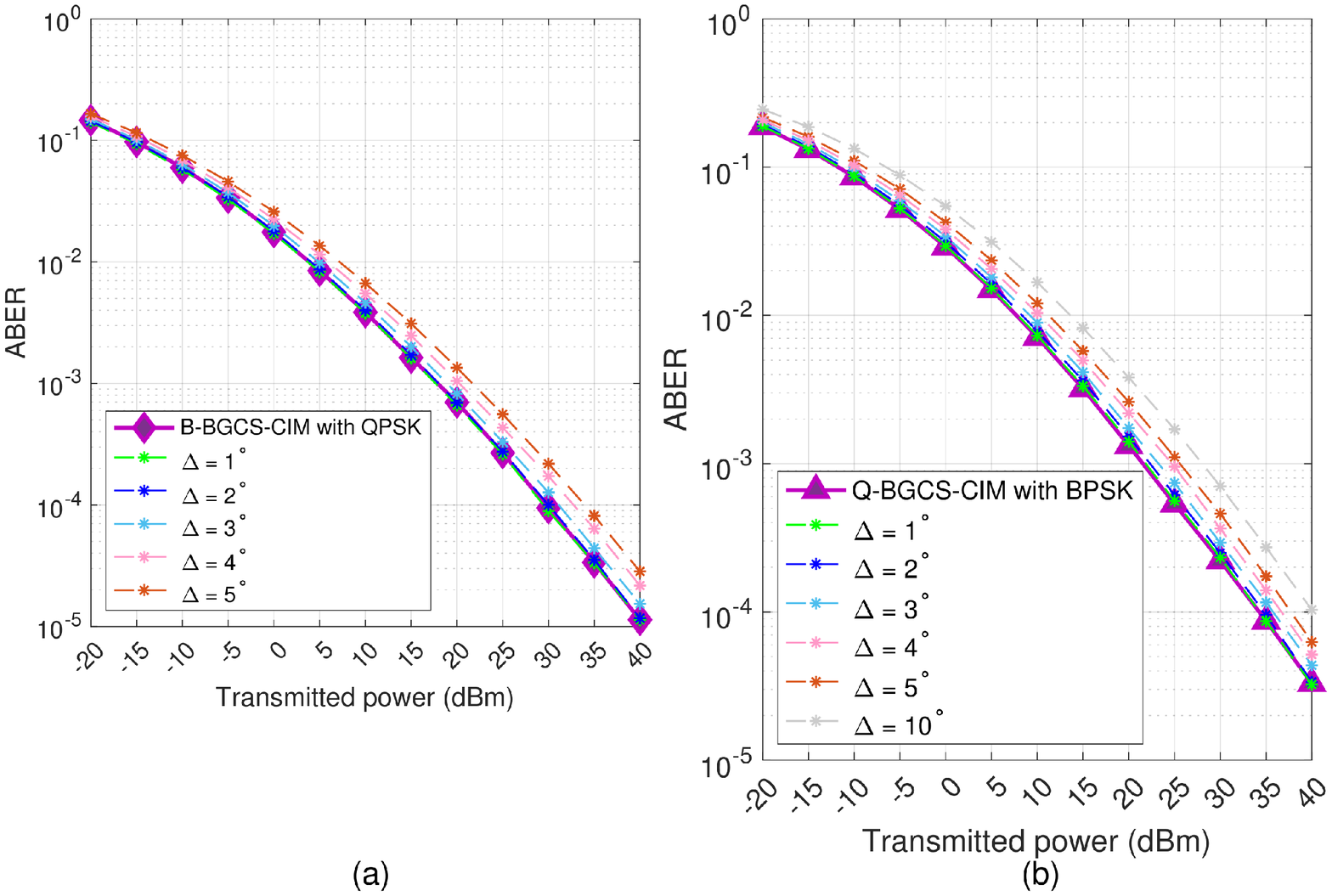}
    \caption{ABER performance degradation with respect to the angular information change for B-BGCS-CIM with QPSK signaling.}
    \label{fig:Robustness}
\end{figure}

We also consider different CIM orders to investigate the effect of Euclidean distance among the indexed clusters. As illustrated in Fig. \ref{fig:CIM_EuclideanDistance}, the increasing number of CIM bits results in degraded ABER performance. This performance loss is due to the shorter Euclidean distance between CIM symbols in Q-BGCS-CIM as compared with B-BGCS-CIM and 8-BGCS-CIM as compared with Q-BGCS-CIM, for the spectral efficiency of 3 and 4 bpcu, respectively. Besides, for higher orders of CIM symbols, spreading error\footnote{When the CIM symbol is detected erroneously, the detector resolves the signal received from the corresponding cluster; consequently, the detected signal does not carry any information unless noise. In other words, in this situation, the detector tries to resolve noise which leads to erroneous $M$-ary symbol detection as well. This phenomenon is called spreading error from the IM symbol to the $M$-ary symbol.} from the CIM symbol to $M$-ary symbol is more likely, which also contributes to the deteriorated ABER performance.

\vspace{-3ex}

\subsection{BGCS-CIM performance in S-V channel model}
\vspace{-0.5ex}
This subsection studies the proposed scheme performance for different numbers of antennas/reflectors and channel sparsity. Fig. \ref{fig:ABER of different array sizes} plots the ABER versus the numbers of antenna elements and reflectors for B-BGCS-CIM with QPSK signaling and transmit power $P = 20$ dBm. With the increasing number of antenna elements and reflectors, the ABER performance improves due to narrower beams, higher array directivity, and lower radiation sidelobes. These characteristics of larger arrays decrease the power leakage into undesirable indexed clusters, resulting in less inter-cluster interference. Hence, the detector can resolve CIM symbols more accurately, and the effect of spreading errors from CIM symbols into $M$-ary symbols is limited. As we mentioned earlier, increasing array size results in higher array directivity, narrower half-power bandwidth (HPBW), and higher side lobe level (SLL). The values for these parameters are summarized in Table \ref{tab:Array charactristics}. Higher array directivity gain results in more SINR, while with narrower HPBW (or beams), most of the transmitted signal concentrates in the desired cluster, resulting in less interference in the undesired clusters. On the other hand, higher SLL also prohibits wasting signal power on the undesired directions. All of these benefits from larger arrays result into ABER performance improvement as illustrated in Fig. \ref{fig:ABER of different array sizes}.

Fig. \ref{fig:ABER_ChannelSparsity3D} illustrates the ABER performance of different IM schemes under the effect of channel sparsity for one-bit IM with QPSK signaling. Fig. \ref{fig:ABER_ChannelSparsity3D}(a) indicates that the proposed BGCS-CIM scheme performs better with increasing paths and clusters since the detector can resolve the received signal more precisely in rich scattered channels. More specifically, with the increasing number of paths in each cluster, BGCS-CIM exploits this enriching in the channel as spatial diversity. On the other hand, increasing the number of clusters gives more opportunity to the BGCS-CIM scheme to index the clusters with higher gains. Consequently, increasing channel sparsity boosts ABER performance in the proposed BGCS-CIM scheme. However, this performance improvement is not observable in the simple CIM and SSM schemes. As explained earlier, in the simple CIM scheme, $\mathbf{G}$ is not involved in the indexing algorithm, while it is shared among all the indexed clusters. Specifically, $\mathbf{G}$ is a LOS channel and does not play any role in the illustrated channel sparsity. Nevertheless, the comparison between Figs. \ref{fig:ABER_ChannelSparsity3D}(a) and \ref{fig:ABER_ChannelSparsity3D}(b) perfectly shows the superiority of BGCS-CIM over simple CIM and indicates that the BGCS-CIM scheme owes its superiority to counting the entire cascaded channel in the codebook construction procedure. On the other hand, Fig. \ref{fig:ABER_ChannelSparsity3D}(c) shows the worst ABER performance of the SSM scheme among its counterparts. Since the SSM scheme indexes the paths without considering the clusters, by increasing channel sparsity, i.e., the number of paths and clusters, there is more potential correlation among the indexed paths.

\begin{table}
    \centering
    \caption{Array characteristics for different array sizes at $0^\circ$ azimuth and $0^\circ$ elevation.}
    \vspace{-2ex}
    \label{tab:Array charactristics}
    \begin{tabular}{c c c}
        \toprule
        \textbf{Array size} & \textbf{Parameter} & \textbf{Values} \\
        \toprule
        \multirow{3}{*}{$8 \times 8$}& Array directivity & $19.74$ dBi  \\
        & HPBW & $12.80^\circ$ Az; $14^\circ$ El\\
        & SLL & $12.80$ dB Az; $12.80$ dB El\\
        \midrule
        \multirow{3}{*}{$12 \times 12$}& Array directivity & $23.34$ dBi  \\
        & HPBW & $8.48^\circ$ Az; $10^\circ$ El\\
        & SLL & $13.06$ dB Az; $13.07$ dB El\\
        \midrule
        \multirow{3}{*}{$20 \times 20$}& Array directivity & $27.85$ dBi  \\
        & HPBW & $5.08^\circ$ Az; $6^\circ$ El\\
        & SLL & $13.19$ dB Az; $13.26$ dB El\\
        \bottomrule
    \end{tabular}
\end{table}

This paper assumes that the channel state information remains the same during a coherent block\footnote{In realistic channel conditions, CSI is divided into two parts, slow and fast time-varying CSI \cite{9086460}. While angular information is slow time-varying information, path gains are fast time-varying CSI. In a realistic channel, angular information changes after several time intervals while the path gains change from one interval to another. However, in this paper, for the sake of simplicity, we ignore this characteristic and assume that both angular information and path gains change from one interval to another \cite{jiang2021generalized, wang2020joint}.}. Nonetheless, in realistic conditions, the angular information slowly changes in time. This change in angular information can cause performance degradation in the proposed scheme. Fig. \ref{fig:Robustness} depicts the performance degradation with respect to the angular information change for B-BGCS-CIM with QPSK signaling. As it is illustrated in Fig. \ref{fig:Robustness}, for $\Delta = 1^{\circ}$ and $\Delta = 2^{\circ}$ changes in the cluster position, the performance is almost the same, while for $\Delta = 5^{\circ}$ changes in the cluster location, the performance degrades by about $4$ dB. It is worth mentioning that the angular spread of each cluster is assumed to be $7.5^{\circ}$; hence, $2^{\circ}$ change in the cluster location is equivalent to more than $26 \%$ of the angular spread, while $5^{\circ}$ is equivalent to more than $66 \%$ of the angular spread. In other words, the BGCS-CIM scheme can withstand about $26 \%$ of the angular spread variation in the cluster location without specific ABER performance degradation.

Ultimately, we analyze the computational complexity of our detector. First, the detector should calculate the effective channel $\mathbf{H_{eff}} = \mathbf{R} \mathbf{\Psi}\mathbf{G}$ which requires $\sim N^2N_t + NN_tN_r$ operations, where $\sim N^2N_t$ operations should be done to calculate $\mathbf{\Psi}\mathbf{G}$ and $\sim NN_tN_r$ operations requires to obtain $\mathbf{H_{eff}}$ completely. Next, the transmit beamformer should be multiplied in the effective channel, i.e., $\mathbf{H_{eff}} \mathbf{f}_t$, which requires $\sim N_r N_t$ operations. The whitening combiner $\mathbf{f}_r$ should be multiplied to build $\mathbf{f}_r^H \mathbf{H}_{eff} \mathbf{f}_t$ which needs $\sim N_r$ operations. This procedure should be repeated $MB$ times where $M$ and $B$ are the orders of $M$-ary and CIM constellations, respectively. Hence, the detector has a complexity order of $\mathcal{O}\big(MB(N^2N_t + N N_t N_r + N_t N_r + N_r)\big)$. By ignoring the lower-order terms $N_tN_r + N_r$, the complexity expression can be simplified as $\mathcal{O} \big(MBNN_t(N + N_r)\big)$.

\vspace{-3ex}
\section{Conclusion} \label{Sec: Conclusion}
\vspace{-1.5ex}
This paper has proposed a novel transmission mechanism for RIS-assisted mmWave systems named BGCS-CIM. We have considered a practical RIS-assisted system model, in which CIM has been delicately integrated to enhance system performance. The simulation results have revealed the superiority of the proposed scheme over the existing benchmarks. We have also derived an analytical upper bound for the proposed scheme, validating our simulations. Maximizing the Euclidean distance among indexed clusters can help the system to further boost performance due to minimizing inter-cluster interference, which we will investigate in our future work.

\vspace{-3ex}
\appendices
\section{Proof of Lemma \ref{lemm: Whitening Filter}}
\label{Appndx: Proof of whitening filter}
\vspace{-1ex}
Different whitening procedures with benefits and detriments are discussed in \cite{kessy2018optimal}. We prefer standardized zero-phase components analysis (ZCA-cor) whitening transformation due to its unique ability to keep the whitened noise $\mathbf{F}_r \mathbf{n}$ maximally similar to the effective noise $\mathbf{W}^H \mathbf{n}$ \cite{kessy2018optimal}. As a result, while we whiten the effective noise, we keep it as similar as possible to its colored version.

\vspace{-3ex}
\section{Proof of Lemma \ref{PEP for correct cluster index detection}}
\label{Proof of PEP for correct cluster index detection}
\vspace{-1ex}
Assume that $\mathbf{x}_1$ and $\mathbf{x}_2$ are two complex vectors with the same dimension. It is straightforward to prove that the following equation is correct.
\begin{equation}\label{eq:square of norm of sum}
    || \mathbf{x}_1 + \mathbf{x}_2 ||^2 = ||\mathbf{x}_1||^2 + ||\mathbf{x}_2||^2 + 2 \Re \{ \mathbf{x}_1^H \mathbf{x}_2 \}.
\end{equation}
Using (\ref{eq:square of norm of sum}), we can update (\ref{eq: CPEP correct CIM detection}) as follows:
\begin{equation}
\begin{split}
    \mathbb{P} (\{ c^*,s^* \} \rightarrow & \{ c^*, \hat{s} \} | \alpha_0, \beta_{1,1}, \dots, \beta_{C_R, L_R}) \\
    & = \mathbb{P} (2 \Re \{ \sqrt{P} \mathbf{n}^H \mathbf{f}_{r,c^*} \mathbf{f}_{r,c^*}^H \mathbf{H_{eff}^{\mathbf{b^*}}} \mathbf{f}_t (s^* - \hat{s}) \} \\
    & + |\sqrt{P} \mathbf{f}_{r,c^*}^H \mathbf{H_{eff}^{\mathbf{b^*}}} \mathbf{f}_t (s^* - \hat{s})|^2 < 0).
\end{split}
\end{equation}
Let us define $\zeta = 2 \Re \{ \sqrt{P}  \mathbf{n}^H \mathbf{f}_{r,c^*} \mathbf{f}_{r,c^*}^H  \mathbf{H_{eff}^{\mathbf{b^*}}} \mathbf{f}_t (s^* - \hat{s}) \} + |\sqrt{P} \mathbf{f}_{r,c^*}^H \mathbf{H_{eff}^{\mathbf{b^*}}} \mathbf{f}_t (s^* - \hat{s})|^2$.
Since $\mathbf{n}^H$ has complex normal distribution with variance $\sigma^2$, its real part has normal distribution with variance $\frac{\sigma^2}{2}$, i.e., $\Re \{ \mathbf{n}^H \} \sim \mathcal{N}(0,\frac{\sigma^2}{2})$. Hence, it is easy to see that $\zeta$ has normal distribution with mean $\mu_{\zeta}$ and variance $\sigma_{\zeta}$, i.e., $\zeta \sim \mathcal{N}(\mu_{\zeta}, \sigma_{\zeta}^2)$, wherein
\begin{equation*}
    \mu_{\zeta} = |\sqrt{P} \mathbf{f}_{r,c^*}^H \mathbf{H_{eff}^{\mathbf{b^*}}} \mathbf{f}_t (s^* - \hat{s})|^2,
\end{equation*}
\begin{equation*}
    \sigma_{\zeta}^2 = 2 \sigma^2 ||\sqrt{P} \mathbf{f}_{r,c^*} \mathbf{f}_{r,c^*}^H \mathbf{H_{eff}^{\mathbf{b^*}}} \mathbf{f}_t (s^* - \hat{s})||^2.
\end{equation*}
Subsequently, we can derive the CPEP as 
\begin{equation}
    \begin{split}
    \mathbb{P} (\{ c^*,s^* \} \rightarrow & \{ c^*, \hat{s} \} | \alpha_0, \beta_{1,1}, \dots, \beta_{C_R, L_R} ) \\
    & = Q \Big( \frac{| \mathbf{f}_{r,c^*}^H \mathbf{H_{eff}^{\mathbf{b^*}}} \mathbf{f}_t (s^* - \hat{s}) |^2}{\sqrt{2} \sigma || \mathbf{f}_{r,c^*} \mathbf{f}_{r,c^*}^H \mathbf{H_{eff}^{\mathbf{b^*}}} \mathbf{f}_t (s^* - \hat{s}) ||} \Big).
\end{split}
\end{equation}
Finally, by taking expectation over enough number of channel realizations, the $\mathrm{UPEP}$ expression can be obtained as in (\ref{eq:correct CIM UPEP}).

\begin{figure*}[t]
    \begin{equation}\label{CEPE proof for errouneous CIM detection}
        \begin{split}
            \mathbb{P} (\{ c^*, & s^* \} \rightarrow  \{ \hat{c}, \hat{s} \} | \alpha_0, \beta_{1,1}, \dots, \beta_{C_R, L_R}) = \mathbb{P} (\chi_1 > \chi_2) =  \int_{0}^{\infty}\int_{x_2}^{\infty} f_{\chi_1,\chi_2}(x_1,x_2) d x_1 d x_2
            = \frac{1}{4} \int_{0}^{\infty} \exp{(-\frac{x_2 + \Lambda}{2})} \\
            & \times \sum_{i = 0}^{\infty} \frac{(\frac{\Lambda x_2}{4})^i}{(i!)^2} \bigg \{ \int_{x_2}^{\infty} \exp{(-\frac{x_1}{2})} d x_1 \bigg \} d x_2 = \frac{\exp{(-\frac{\Lambda}{2})}}{2} \int_{0}^{\infty} \exp{(-x_2)} \sum_{i = 0}^{\infty} \frac{(\frac{\Lambda x_2}{4})^i}{(i!)^2} d x_2 =  \frac{\exp{(-\frac{\Lambda}{2})}}{2} \sum_{i = 0}^{\infty} \bigg \{ \frac{(\frac{\Lambda}{4})^i}{(i!)^2} \\
            & \times \int_{0}^{\infty} x_2^i \exp{(-x_2)} d x_2   \bigg \} = \frac{\exp{(-\frac{\Lambda}{2})}}{2}  \sum_{i = 0}^{\infty} \frac{(\frac{\Lambda}{4})^i}{i!} = \frac{\exp{(-\frac{\Lambda}{4})}}{2} = \frac{1}{2} \exp{\Big( - \frac{|P\mathbf{f}_{r,\hat{c}}^H  (\mathbf{H_{eff}^{b^*}} s^* - \mathbf{H_{eff}^{\hat{b}}} \hat{s} )  \mathbf{f}_t |^2}{2 \sigma^2} \Big)}.
        \end{split}
    \end{equation}
    \hrulefill
    \vspace{-4ex}
\end{figure*}

\vspace{-3ex}
\section{Proof of Lemma \ref{PEP for erroneous cluster index detection}}
\label{Proof of PEP for erroneous cluster index detection}
\vspace{-1ex}
Let us define 
\begin{equation}
    \chi_1 = \frac{|\mathbf{f}_{r,c^*}^H \mathbf{n}|^2}{\sigma^2/2},
\end{equation}
\begin{equation}
    \chi_2 = \frac{|\mathbf{f}_{r,\hat{c}}^H \mathbf{n}
    + \sqrt{P}\mathbf{f}_{r,\hat{c}}^H (\mathbf{H_{eff}^{b^*}} s^* - \mathbf{H_{eff}^{\hat{b}}} \hat{s} ) \mathbf{f}_t|^2}{\sigma^2/2},
\end{equation}
where $\chi_1$ and $\chi_2$ are central and non-central chi-squared random variables with two degrees of freedom, respectively. Notice that, $\chi_2$ has a non-centrality parameter as:
\begin{equation}
    \Lambda = \frac{2|\sqrt{P}\mathbf{f}_{r,\hat{c}}^H (\mathbf{H_{eff}^{b^*}} s^* - \mathbf{H_{eff}^{\hat{b}}} \hat{s} ) \mathbf{f}_t|^2}{\sigma^2}.
\end{equation}
The PDF of central and non-central random variable with $i$ degrees of freedom and non-centrality parameter of $\Lambda$ is as follows, respectively:

\begin{equation}
    f_{\chi_1}(x_1) = \frac{1}{2^{\frac{i}{2}} \Gamma(\frac{i}{2})} x_1 ^ {\frac{i}{2} - 1} \exp{(-\frac{x_1}{2})},
\end{equation}
\begin{equation}
    f_{\chi_2}(x_2) = \frac{1}{2} \exp{(-\frac{x_2 + \Lambda}{2})} (\frac{x_2}{\Lambda})^{\frac{k}{4} - \frac{1}{2}} I_{\frac{k}{2} - 1}(\sqrt{\Lambda x_2}),
\end{equation}
where $I_{\nu}(y)$ is the Bessel function of the first kind and it is defined as:
\begin{equation}
    I_{\nu}(y) = (\frac{y}{2})^{\nu} \sum_{i = 0}^{\infty} \frac{{(y^2/4)}^i}{i! \Gamma(\nu + i + 1)}.
\end{equation}

Accordingly, for two degrees of freedom, the PDF of $\chi_1$ and $\chi_2$ can be obtained as follows, respectively:
\begin{equation}\label{PDF of chi1}
    f_{\chi_1}(x_1) = \frac{1}{2} \exp{(-\frac{x_1}{2})}.
\end{equation}
\begin{equation}\label{PDF of chi2}
    f_{\chi_2}(x_2) = \frac{1}{2} \exp{(-\frac{x_2 + \Lambda}{2})} \sum_{i = 0}^{\infty} \frac{(\frac{\Lambda x_2}{4})^i}{(i!)^2}.
\end{equation}

Thanks to the whitening filter, $\chi_1$ and $\chi_2$ are independent; therefore, we can constitute their joint PDF as follows:
\begin{equation}
    \begin{split}
       f_{\chi_1,\chi_2}(x_1,x_2) & = f_{\chi_1}(x_1) f_{\chi_2}(x_2)\\
       & = \frac{1}{4} \exp{(-\frac{x_1 + x_2 + \Lambda}{2})} \sum_{i = 0}^{\infty} \frac{(\frac{\Lambda x_2}{4})^i}{(i!)^2}. 
    \end{split}
\end{equation}
Subsequently, we can derive the CPEP expression for the case of erroneous cluster index detection as (\ref{CEPE proof for errouneous CIM detection}) on top of this page. It is worth mentioning that during solving process of (\ref{CEPE proof for errouneous CIM detection}), we used Maclaurin series for $e^x$, i.e., $\sum_{i=0}^{\infty}\frac{x^i}{i!} = e^x$, and $\int_{0}^{\infty}x^i e^{-x} dx = i!$; the latter can be obtained by extending the factorial to a continuous function using definition of Gamma function.
Ultimately, by taking expectation over enough channel realizations, we can obtain the $\mathrm{UPEP}$ as (\ref{UPEP for errouneous index detection}).



\ifCLASSOPTIONcaptionsoff
  \newpage
\fi

\vspace{-2.5ex}

\bibliographystyle{ieeetr}
\bibliography{reference.bib}

\begin{thebibliography}{10}

\bibitem{9386246}
K.~Ntontin, A.-A.~A. Boulogeorgos, D.~G. Selimis, F.~I. Lazarakis, A.~Alexiou,
  and S.~Chatzinotas, ``Reconfigurable intelligent surface optimal placement in
  millimeter-wave networks,'' {\em IEEE open j. Commun. Soc.}, vol.~2,
  pp.~704--718, Mar. 2021.

\bibitem{wang2020joint}
P.~Wang, J.~Fang, L.~Dai, and H.~Li, ``Joint transceiver and large intelligent
  surface design for massive {MIMO} {mmWave} systems,'' {\em IEEE Trans. Wirel.
  Commun.}, vol.~20, pp.~1052--1064, Oct. 2020.

\bibitem{wang2020intelligent}
P.~Wang, J.~Fang, X.~Yuan, Z.~Chen, and H.~Li, ``Intelligent reflecting
  surface-assisted millimeter wave communications: {Joint} active and passive
  precoding design,'' {\em IEEE Trans. Veh. Technol.}, vol.~69,
  pp.~14960--14973, Oct. 2020.

\bibitem{ying2020gmd}
K.~Ying, Z.~Gao, S.~Lyu, Y.~Wu, H.~Wang, and M.-S. Alouini, ``{GMD}-based
  hybrid beamforming for large reconfigurable intelligent surface assisted
  millimeter-wave massive {MIMO},'' {\em IEEE Access}, vol.~8,
  pp.~19530--19539, Jan. 2020.

\bibitem{basar2019wireless}
E.~Basar, M.~Di~Renzo, J.~De~Rosny, M.~Debbah, M.-S. Alouini, and R.~Zhang,
  ``Wireless communications through reconfigurable intelligent surfaces,'' {\em
  IEEE Access}, vol.~7, pp.~116753--116773, Aug. 2019.

\bibitem{basar2021reconfigurable}
E.~Basar and I.~Yildirim, ``Reconfigurable intelligent surfaces for future
  wireless networks: {A} channel modeling perspective,'' {\em IEEE Wirel.
  Commun.}, vol.~28, pp.~108--114, Apr. 2021.

\bibitem{guo2021joint}
B.~Guo, R.~Li, and M.~Tao, ``Joint design of hybrid beamforming and phase
  shifts in {RIS}-aided {mmWave} communication systems,'' in {\em Proc. IEEE
  Wireless Commun. Netw. Conf. (WCNC)}, pp.~1--6, Mar. 2021.

\bibitem{9687593}
R.~Li, B.~Guo, M.~Tao, Y.-F. Liu, and W.~Yu, ``Joint design of hybrid
  beamforming and reflection coefficients in {RIS}-aided {mmWave} {MIMO}
  systems,'' {\em IEEE Trans. Commun.}, vol.~70, pp.~2404--2416, Jan. 2022.

\bibitem{jiang2021generalized}
Y.~Jiang, H.~Hu, S.~Yang, J.~Zhang, and J.~Zhang, ``Generalized {3-D} spatial
  scattering modulation,'' {\em IEEE Trans. Wirel. Commun.}, vol.~21,
  pp.~1570--1585, Aug. 2021.

\bibitem{basar2016index}
E.~Basar, ``Index modulation techniques for {5G} wireless networks,'' {\em IEEE
  Commun. Mag.}, vol.~54, pp.~168--175, Jul. 2016.

\bibitem{basar2017index}
E.~Basar, M.~Wen, R.~Mesleh, M.~Di~Renzo, Y.~Xiao, and H.~Haas, ``Index
  modulation techniques for next-generation wireless networks,'' {\em IEEE
  Access}, vol.~5, pp.~16693--16746, Aug. 2017.

\bibitem{ding2017spatial}
Y.~Ding, K.~J. Kim, T.~Koike-Akino, M.~Pajovic, P.~Wang, and P.~Orlik,
  ``Spatial scattering modulation for uplink millimeter-wave systems,'' {\em
  IEEE Commun. Lett.}, vol.~21, pp.~1493--1496, Mar. 2017.

\bibitem{raeisi2022cluster}
M.~Raeisi, A.~Koc, E.~Basar, and T.~Le-Ngoc, ``Cluster index modulation for
  {mmWave} communication systems,'' {\em Front. Comms. Net 2: 803007. doi:
  10.3389/frcmn}, Feb. 2022.

\bibitem{liu2020machine}
H.~Liu, S.~Lu, M.~El-Hajjar, and L.-L. Yang, ``Machine learning assisted
  adaptive index modulation for {mmWave} communications,'' {\em IEEE open j.
  Commun. Soc.}, vol.~1, pp.~1425--1441, Sep. 2020.

\bibitem{yildirim2020modeling}
I.~Yildirim, A.~Uyrus, and E.~Basar, ``Modeling and analysis of reconfigurable
  intelligent surfaces for indoor and outdoor applications in future wireless
  networks,'' {\em IEEE Trans. Commun.}, vol.~69, pp.~1290--1301, Nov. 2020.

\bibitem{tu2018generalized}
Y.~Tu, L.~Gui, Q.~Qin, L.~Zhang, J.~Xiong, and M.~Yang, ``Generalized spatial
  scattering modulation for uplink millimeter wave {MIMO} system,'' in {\em
  Proc. IEEE Int. Conf. Commun. China (ICCC)}, pp.~22--27, Feb. 2018.

\bibitem{gopi2020intelligent}
S.~Gopi, S.~Kalyani, and L.~Hanzo, ``Intelligent reflecting surface assisted
  beam index-modulation for millimeter wave communication,'' {\em IEEE Trans.
  Wirel. Commun.}, vol.~20, pp.~983--996, Oct. 2020.

\bibitem{li2019polarized}
Q.~Li, K.~J. Kim, S.~Ruan, L.~Yuan, L.~Yang, and J.~Zhang, ``Polarized spatial
  scattering modulation,'' {\em IEEE Commun. Lett.}, vol.~23, pp.~2252--2256,
  Sep. 2019.

\bibitem{ruan2019diversity}
S.~Ruan, B.~Hu, K.~J. Kim, Q.~Li, L.~Yuan, L.~Jin, and J.~Zhang, ``Diversity
  analysis for spatial scattering modulation in millimeter wave {MIMO}
  system,'' in {\em Proc. 11th Int. Conf. Wireless Commun. Signal Process.
  (WCSP)}, pp.~1--5, Dec. 2019.

\bibitem{9685434}
J.~Hu, H.~Yin, and E.~Björnson, ``{MmWave} {MIMO} communication with
  semi-passive {RIS}: {A} low-complexity channel estimation scheme,'' in {\em
  2021 IEEE Global Communications Conference (GLOBECOM)}, pp.~01--06, Dec.
  2021.

\bibitem{9955484}
R.~Liu, J.~Dou, P.~Li, J.~Wu, and Y.~Cui, ``Simulation and field trial results
  of reconfigurable intelligent surfaces in {5G} networks,'' {\em IEEE Access},
  vol.~10, pp.~122786--122795, Nov. 2022.

\bibitem{liu2020deep}
S.~Liu, Z.~Gao, J.~Zhang, M.~Di~Renzo, and M.-S. Alouini, ``Deep denoising
  neural network assisted compressive channel estimation for {mmWave}
  intelligent reflecting surfaces,'' {\em IEEE Trans. Veh. Technol.}, vol.~69,
  pp.~9223--9228, Jun. 2020.

\bibitem{9103231}
P.~Wang, J.~Fang, H.~Duan, and H.~Li, ``{Compressed} channel estimation for
  intelligent reflecting surface-assisted millimeter wave systems,'' {\em IEEE
  Signal Process. Lett.}, vol.~27, pp.~905--909, May. 2020.

\bibitem{wang2019towards}
J.~Wang, L.~He, and J.~Song, ``Towards higher spectral efficiency: {Spatial}
  path index modulation improves millimeter-wave hybrid beamforming,'' {\em
  IEEE J. Sel. Top. Signal Process.}, vol.~13, pp.~1348--1359, May. 2019.

\bibitem{akdeniz2014millimeter}
M.~R. Akdeniz, Y.~Liu, M.~K. Samimi, S.~Sun, S.~Rangan, T.~S. Rappaport, and
  E.~Erkip, ``Millimeter wave channel modeling and cellular capacity
  evaluation,'' {\em IEEE J. Sel. Areas Commun.}, vol.~32, pp.~1164--1179, Jun.
  2014.

\bibitem{mahmood20223}
M.~Mahmood, A.~Koc, and T.~Le-Ngoc, ``{3-D} antenna array structures for
  millimeter wave multi-user massive {MIMO} hybrid precoder design: A
  performance comparison,'' {\em IEEE Commun. Lett.}, vol.~26, pp.~1393--1397,
  Mar. 2022.

\bibitem{el2014spatially}
O.~El~Ayach, S.~Rajagopal, S.~Abu-Surra, Z.~Pi, and R.~W. Heath, ``Spatially
  sparse precoding in millimeter wave {MIMO} systems,'' {\em IEEE Trans. Wirel.
  Commun.}, vol.~13, pp.~1499--1513, Jan. 2014.

\bibitem{9390210}
A.~Koc and T.~Le-Ngoc, ``Full-duplex {mmWave} massive {MIMO} systems: A joint
  hybrid precoding/combining and self-interference cancellation design,'' {\em
  IEEE open j. Commun. Soc.}, vol.~2, pp.~754--774, Mar. 2021.

\bibitem{etsitr138901}
{\em {5G}: Study on Channel Model for Frequencies From 0.5 to 100 {GHz}}.
\newblock document 3GPP TR 38.901, Ver. 16.1.0, Nov. 2020.

\bibitem{ning2021terahertz}
B.~Ning, Z.~Chen, W.~Chen, Y.~Du, and J.~Fang, ``Terahertz multi-user massive
  {MIMO} with intelligent reflecting surface: Beam training and hybrid
  beamforming,'' {\em IEEE Trans. Veh. Technol.}, vol.~70, pp.~1376--1393, Jan.
  2021.

\bibitem{9086460}
A.~Koc, A.~Masmoudi, and T.~Le-Ngoc, ``{3D} angular-based hybrid precoding and
  user grouping for uniform rectangular arrays in massive {MU-MIMO} systems,''
  {\em IEEE Access}, vol.~8, pp.~84689--84712, May. 2020.

\bibitem{kessy2018optimal}
A.~Kessy, A.~Lewin, and K.~Strimmer, ``Optimal whitening and decorrelation,''
  {\em Amer. Stat.}, vol.~72, pp.~309--314, Jan. 2018.

\end{thebibliography}




\end{document}